%% file: main.tex
\documentclass[a4paper,USenglish,cleveref, autoref, authorcolumns]{lipics-v2019}

\usepackage[utf8]{inputenc}

\title{When Lipschitz Walks Your Dog:\newline Algorithm Engineering of the\newline Discrete Fréchet Distance under Translation}

\author{Karl Bringmann}{Saarland University and Max Planck Insitute for Informatics,  Saarland Informatics Campus, Saarbrücken, Germany}{kbringma@mpi-inf.mpg.de}{}{This work is part of the project TIPEA that has received funding from the European Research Council (ERC) under the European Unions Horizon 2020 research and innovation programme (grant agreement No. 850979).}
\author{Marvin Künnemann}{Max Planck Insitute for Informatics,  Saarland Informatics Campus, Saarbrücken, Germany}{marvin@mpi-inf.mpg.de}{}{}
\author{André Nusser}{Saarbrücken Graduate School of Computer Science and Max Planck Insitute for Informatics,  Saarland Informatics Campus, Saarbrücken, Germany}{anusser@mpi-inf.mpg.de}{}{}

\titlerunning{Algorithm Engineering of the Discrete Fréchet Distance under Translation} 

\authorrunning{K. Bringmann, M. Künnemann, A. Nusser} 

\Copyright{Karl Bringmann, Marvin Künnemann, and André Nusser} 



\category{} 

\relatedversion{} 

\supplement{\url{https://gitlab.com/anusser/frechet_distance_under_translation}}

\acknowledgements{We thank Andreas Karrenbauer for helpful discussions.}

\nolinenumbers 

\hideLIPIcs  

\usepackage{amsmath,amssymb,amsfonts,amsthm}
\usepackage{mathtools}
\usepackage{nicefrac}
\usepackage{csquotes}
\usepackage{xspace}
\usepackage{color}
\usepackage{pdflscape}
\usepackage{subcaption}

\newtheorem{observation}{Observation}

\newcommand{\ANDRE}[1]{\noindent\textcolor{red}{\textbf{TODO (André):} #1}}
\newcommand{\MARVIN}[1]{\noindent\textcolor{green}{\textbf{TODO (Marvin):} #1}}
\newcommand{\KARL}[1]{\noindent\textcolor{blue}{\textbf{TODO (Karl):} #1}}
\newcommand{\andre}[1]{\ANDRE{#1}}
\newcommand{\marvin}[1]{\MARVIN{#1}}
\newcommand{\karl}[1]{\KARL{#1}}
\renewcommand{\andre}[1]{}
\renewcommand{\marvin}[1]{}
\renewcommand{\karl}[1]{}
\renewcommand{\ANDRE}[1]{}
\renewcommand{\MARVIN}[1]{}
\renewcommand{\KARL}[1]{}

\usepackage{algorithm}
\usepackage[noend]{algpseudocode}

\usepackage{booktabs,makecell}
\newcolumntype{L}{>{\raggedright\arraybackslash}X}

\newcommand{\norm}[1]{\left\lVert #1 \right\rVert}

\newcommand{\Oh}{\ensuremath{\mathcal{O}}}
\newcommand{\dF}{d_{F}}
\newcommand{\sigspatial}{\textsc{Sigspatial}\xspace}
\newcommand{\characters}{\textsc{Characters}\xspace}

\newcommand{\taustart}{\tau_{\mathrm{start}}}
\newcommand{\tauend}{\tau_{\mathrm{end}}}
\newcommand{\deltastart}{\delta_{\mathrm{start}}}
\newcommand{\deltaend}{\delta_{\mathrm{end}}}
\newcommand{\tauopt}{\tau^*}
\newcommand{\deltaopt}{\delta^*}
\newcommand{\deltaLB}{\delta_{\mathrm{LB}}}
\newcommand{\deltaUB}{\delta_{\mathrm{UB}}}

\newcommand{\dtransF}{d_{\text{trans-}F}}
\newcommand{\eps}{\epsilon}

\newcommand{\lmf}{LMF\xspace}

\newcommand{\Nsamples}{N_{\mathrm{samples}}}
\newcommand{\sizeparameter}{\gamma_{\mathrm{size}}}
\newcommand{\depthparameter}{\gamma_{\mathrm{depth}}}

\newcommand{\arr}{\mathcal{A}}
\newcommand{\R}{\mathbb{R}}

\begin{document}

\maketitle
\input{abstract.tex}


\newpage
\input{trunk/introduction.tex}

\input{trunk/preliminaries.tex}
\input{trunk/our_approach.tex}
\input{trunk/decider.tex}

\input{trunk/valComputation.tex}
\input{trunk/experiments.tex}

\input{trunk/conclusion.tex}

\bibliography{trunk/biblio}

\end{document}

%% file: abstract.tex

\begin{abstract}
Consider the natural question of how to measure the similarity of curves in the plane by a quantity that is \emph{invariant under translations} of the curves. Such a measure is justified whenever we aim to quantify the similarity of the curves' \emph{shapes} rather than their positioning in the plane, e.g., to compare the similarity of handwritten characters.
Perhaps the most natural such notion is the (discrete) \emph{Fréchet distance under translation}. Unfortunately, the algorithmic literature on this problem yields a very pessimistic view: On polygonal curves with $n$ vertices, the fastest algorithm runs in time $\Oh(n^{4.667})$ and cannot be improved below $n^{4-o(1)}$ unless the Strong Exponential Time Hypothesis fails. Can we still obtain an implementation that is efficient on realistic datasets? 

Spurred by the surprising performance of recent implementations for the Fréchet distance, we perform algorithm engineering for the Fréchet distance under translation. Our solution combines fast, but inexact tools from continuous optimization (specifically, branch-and-bound algorithms for global Lipschitz optimization) with exact, but expensive algorithms from computational geometry (specifically, problem-specific algorithms based on an arrangement construction).
We combine these two ingredients to obtain an \emph{exact decision} algorithm for the Fréchet distance under translation. 
For the related task of computing the distance \emph{value} up to a desired precision, we engineer and compare different methods. On a benchmark set involving handwritten characters and route trajectories, our implementation answers a typical query for either task in the range of a few milliseconds up to a second on standard desktop hardware.   

We believe that our implementation will enable, for the first time, the use of the Fréchet distance under translation in applications, whereas previous algorithmic approaches would have been computationally infeasible. 
Furthermore, we hope that our combination of continuous optimization and computational geometry will inspire similar approaches for further algorithmic questions.

\end{abstract}

%% file: trunk/introduction.tex

\section{Introduction}

Consider the following natural task: Given two handwritings of (the same or different) characters, represented as polygonal curves $\pi, \sigma$ in the plane, determine how similar they are.
To measure the similarity of two such curves, several distance notions could be used, where the most popular measure in computational geometry is given by the \emph{Fréchet distance} $\dF(\pi, \sigma)$:  Intuitively, we imagine a dog walking on $\pi$ and its owner walking on $\sigma$, and define $\dF(\pi, \sigma)$ as the \emph{shortest leash length} required to connect the dog to its owner while both walk along their curves (only forward, but at arbitrarily and independently variable speeds). In this paper, we focus on the \emph{discrete} version, in which dog and owner do not continuously walk along the curves, but jump from vertex to vertex.\footnote{We give a precise definition in Section~\ref{sec:prelim}.} As a fundamental curve similarity notion that takes into account the \emph{sequence} of the points of the curves (rather than simply the set of points, as in the simpler notion of the Hausdorff distance), the discrete Fréchet distance and variants have received considerable attention from the computational geometry community, see, e.g.~\cite{AltG95, EiterM94, BuchinBW09, DriemelHPW12, AgarwalBAKS14, Bringmann14, BuchinBMM17,BuchinOS19}. While the fastest known algorithms take time $n^{2\pm o(1)}$ on polygonal curves with at most $n$ vertices~\cite{AltG95, EiterM94, AgarwalBAKS14,BuchinBMM17}---which is best possible under the Strong Exponential Time Hypothesis~\cite{Bringmann14}---a recent line of research~\cite{sigspatial1,sigspatial2,sigspatial3,BringmannKN19socg} gives fast implementations for practical input curves.

In the setting of handwritten characters, one would expect our notion of similarity to be \emph{invariant under translations} of the curves; after all, translating one character in the plane while fixing the position of the other should not affect their similarity. In this sense, the original Fréchet distance seems inadequate, as it does not satisfy translation invariance. However, we may canonically define a translation-invariant adaptation as the minimum Fréchet distance between $\pi$ and any translation of $\sigma$, yielding the \emph{Fréchet distance under translation}. Note that beyond computing the similarity of handwritten characters, this measure is generally applicable whenever our intuitive notion of similarity is not affected by translations, such as recognition of movement patterns\footnote{One may argue that the similarity of movement patterns also depends on the speed/velocity of the motion. In principle, we can also incorporate such information into any Fréchet-distance-based measure by introducing an additional dimension.}. In some settings, we would expect our notion to additionally be scaling- or rotation-invariant; however, this is beyond the scope of this paper, as already the Fréchet distance under translation presents previously unresolved challenges.

Can we compute the Fréchet distance under translation quickly? 
The existing theoretical work yields a rather pessimistic outlook: For the discrete Fréchet distance under translation in the plane, the currently fastest algorithm runs in time $\Oh(n^{4.667})$, and any algorithm requires time $n^{4-o(1)}$ under the Strong Exponential Time Hypothesis~\cite{BringmannKN19soda}.  These high polynomial bounds appear prohibitive in practice, and have likely impeded algorithmic uses of this similarity measure. (For the continuous analogue, the situation appears even worse, as the fastest algorithm has a significantly higher worst-case bound of $\Oh(n^8\log n)$; we thus solely consider the discrete version in this work.) Given the surprising performance of recent Fréchet distance implementations on realistic curves~\cite{giscup17, BringmannKN19socg}, can we still hope for faster algorithms on realistic inputs also for its translation-invariant version? 

\subparagraph*{Our problem.} Towards making the Fréchet distance under translation applicable for practical applications, we engineer a fast implementation and analyze it empirically on realistic input sets. 
%
%
Perhaps surprisingly, our fastest solution for the problem combines inexact \emph{continuous optimization} techniques with an exact, but expensive problem-specific approach from computational geometry to obtain an \emph{exact decision} algorithm. We discuss our approach in Section~\ref{sec:approach} and present the details of our decision algorithm in Section~\ref{sec:decider}. 
We develop our approach also for the related, but different task to compute the distance value up to a given precision in Section~\ref{sec:valueComputation}, and evaluate our solutions for both settings in comparison to baseline approaches in Section~\ref{sec:experiments}.

\subparagraph*{Further related work.} 
Variations of the distance measure studied in this paper arise by choosing (1) the discrete or continuous Fréchet distance, (2) the dimension $d$ of the ambient Euclidean space, 
and (3) a class of transformations, e.g., translations, rotations, scaling, or arbitrary linear transformations. 
A detailed treatment of algorithms for this class of distance measures can be found in~\cite{wenk2002phd}.
The earliest example of a problem in this class is the continuous Fréchet distance under translations in dimension $d=2$, which was introduced by Alt et al.~\cite{alt2001matching} together with an $\Oh(n^8 \log n)$-time algorithm.

In this paper we focus on the discrete Fréchet distance under translation in the plane. This problem was first studied by Mosig and Clausen~\cite{MOSIG2005113}, who gave an $\Oh(n^4)$ algorithm for approximating the discrete Fréchet distance under rigid motions. Subsequently, Jiang et al.~\cite{JiangXZ08} presented an $\Oh(n^6 \log n)$-time algorithm for the exact Fréchet distance under translation.
Their running time was improved by Ben Avraham et al.\ to $\Oh(n^5 \log n)$~\cite{BenAvrahamKS15}, and then by Bringmann et al.\ to $\Oh(n^{4.667})$~\cite{BringmannKN19soda}. 
A conditional lower bound of $n^{4-o(1)}$ can be found in~\cite{BringmannKN19soda}.

Algorithm engineering efforts for the Fréchet distance were initiated by the SIGSPATIAL GIS Cup 2017~\cite{giscup17}, 
where the task was to implement a nearest neighbor data structure for curves under the Fréchet distance; 
see~\cite{sigspatial1,sigspatial2,sigspatial3} for the top three submissions. 
The currently fastest implementation of the Fréchet distance is due to Bringmann et al.~\cite{BringmannKN19socg}. 
Further recent directions of Fréchet-related algorithm engineering include k-means clustering of trajectories~\cite{BuchinDLN19}
and locality sensitive hashing of trajectories~\cite{CeccarelloDS19}.

%

%% file: trunk/preliminaries.tex
\section{Preliminaries}
\label{sec:prelim}

Throughout the paper, we consider the Euclidean plane and denote the Euclidean norm by~$\norm{\cdot}$.
A \emph{polygonal curve} $\pi$ is a sequence $\pi = (\pi_1,\dots, \pi_n)$ of vertices $\pi_i \in \R^2$. For any $\tau \in \R^2$, we write $\pi + \tau$ for the translated curve $(\pi_1 + \tau, \dots, \pi_n + \tau)$. 

For any curves $\pi = (\pi_1, \dots, \pi_n), \sigma= (\sigma_1,\dots, \sigma_m)$, we define their \emph{discrete Fréchet distance} as follows. A \emph{traversal} is a sequence $T = (\, (p_1,s_1), \dots, (p_t, s_t)\, )$ of pairs $(p_i,s_i) \in [n]\times[m]$ such that $(p_1,s_1) = (1,1)$, $(p_t,s_t) = (n,m)$ and $(p_{i+1},s_{i+1})\in \{(p_i+1, s_i), (p_i, s_i+1), (p_i+1, s_i+1)\}$ for all $1\le i < t$. The width of a traversal is $\max_{i=1, ..., |T|} \norm{\pi_{p_i}-\sigma_{s_i}}$. The discrete Fréchet distance is then defined as the smallest width over all traversals, i.e., 
\[ \dF(\pi, \sigma) \coloneqq \min_{\text{traversal } T} \max_{i=1, ..., |T|} \norm{\pi_{p_i} - \sigma_{s_i}}.\]
As we only consider the discrete Fréchet distance in this paper, we drop ``discrete'' in the remainder. To avoid confusion, we also refer to it as the \emph{fixed-translation} Fréchet distance.

As the canonically translation-invariant variant of the discrete Fréchet distance, we define the \emph{discrete Fréchet distance under translation} as $\dtransF(\pi, \sigma) \coloneqq \min_{\tau\in \R^2} \dF(\pi, \sigma+\tau)$. 
We typically view the problem as a two-dimensional optimization problem with objective function $f(\tau) \coloneqq \dF(\pi, \sigma + \tau)$. Specifically, we consider the task to decide \emph{$\min_{\tau \in \R^2} f(\tau) \le \delta$?} (\emph{exact decider}) or to return a value in the range $[(1-\eps)\min_{\tau \in \R^2} f(\tau), (1+\eps) \min_{\tau \in \R^2} f(\tau)]$ (\emph{approximate value computation}, multiplicative version). In fact, for implementation reasons (see Section~\ref{sec:valueComputation} for the details),
  our implementation returns a value in $[\min_{\tau \in \R^2} f(\tau) - \eps, \min_{\tau \in \R^2} f(\tau) + \eps]$ (\emph{approximate value computation}, additive version) using a straightforward adaptation of our approach.

Apart from a black-box Fréchet oracle answering decision queries \emph{$\dF(\pi, \sigma + \tau)  \le \delta$?}, our algorithms only exploit the following simple properties:

\begin{observation}[Lipschitz property]\label{obs:Lipschitz}
The objective function $f$ is 1-Lipschitz, i.e., $|f(\tau) - f(\tau + \tau')| \le \norm{\tau'}$.
\end{observation}
\begin{proof}
Note that for any $\pi_i,\sigma_j,\tau,\tau' \in \R^2$, we have
\[
\left|\norm{\pi_i-(\sigma_j+\tau+\tau')} - \norm{\pi_i-(\sigma_j+\tau)}\right| \le \norm{\tau'}
\]
by triangle inequality. Thus, the widths of any traversal $T$ for $\pi, \sigma+\tau$ and $\pi, \sigma+\tau+\tau'$ differ by at most $\norm{\tau'}$, which immediately yields the observation.
\end{proof}

We obtain a simple $2$-approximation of the Fréchet distance under translation as follows.

\begin{observation}
Let $\taustart \coloneqq \pi_1 - \sigma_1$ be the translation of $\sigma$ that aligns the first points of $\pi$ and $\sigma$. Then $\dF(\pi, \sigma + \taustart) \le 2\cdot \dtransF(\pi, \sigma)$.

Analogously, for $\tauend \coloneqq \pi_n - \sigma_m$, we have $\dF(\pi, \sigma + \tauend) \le 2\cdot \dtransF(\pi, \sigma)$. 
\end{observation}
\begin{proof}
Let $\deltaopt \coloneqq \dtransF(\pi, \sigma)$ and let $\tauopt$ be such that $\dF(\pi, \sigma + \tauopt) = \deltaopt$, which implies in particular that $\norm{\pi_1 - (\sigma_1 + \tauopt)} \le \deltaopt$. Thus, $\norm{\taustart-\tauopt} = \norm{\pi_1 - (\sigma_1 + \tauopt)} \le \deltaopt$. Thus by Observation~\ref{obs:Lipschitz}, we obtain $\dF(\pi, \sigma + \taustart) \le \dF(\pi, \sigma + \tauopt) + \deltaopt = 2\deltaopt$. 
\end{proof}

Note that the above observation gives a formal guarantee of a simple heuristic: translate the curves such that the start points match, and compute the corresponding fixed-translation Fr\'echet distance. Unfortunately, this worst-case guarantee is tight\footnote{To see this, take any segment in the plane and let $\pi$ traverse it in one direction, and $\sigma$ in the other. Then the heuristic would return as estimate two times the segment length (the distance of the translated end points), while the optimal translation aligns the segments and achieves the segment length as Fr\'echet distance.} -- a correspondingly large discrepancy is also observed on our data sets.

%% file: trunk/our_approach.tex

\section{Our Approach: Lipschitz meets Fréchet}
\label{sec:approach}

To obtain a fast exact decider, we approach the problem from two different angles: First, we review previous problem-specific approaches to the Fréchet distance under translation, all relying on the construction of an arrangement of circles as an essential tool from computational geometry. Second, we cast the problem into the framework of global Lipschitz optimization with its rich literature on fast, numerical solutions. In isolation, both approaches are inadequate to obtain a fast, exact decider (as the arrangement can be prohibitively large even for realistic data sets, and black-box Lipschitz optimization methods cannot return an exact optimum).  We then describe how to combine both approaches to obtain a fast implementation of an exact decider for the discrete Fréchet distance under translation in the plane. We evaluate our approach, including comparisons to (typically computationally infeasible) baseline approaches, on a data set that we craft from sets of handwritten character and (synthetic) GPS trajectories used in the ACM SIGSPATIAL GIS Cup 2017~\cite{characters_dataset, sigspatial_dataset}. We believe that our approach will inspire similar combinations of fast, inexact methods from continuous optimization with expensive, but exact approaches from computational geometry also in other contexts.

\subsection{View I: Arrangement-based Algorithms} \label{sec:arrView}

Previous algorithms for the Fréchet distance under translation in the plane work as follows. Given two polygonal curves $\pi, \sigma$ and a decision distance $\delta$, consider the set of circles
\[
	\mathcal{C} \coloneqq \{ C_\delta(\pi_i - \sigma_j) \mid \pi_i \in \pi,\; \sigma_j \in \sigma \},
\]
where $C_r(p)$ denotes the circle of radius $r \in \mathbb{R}$ around $p \in \mathbb{R}^2$.
Define the arrangement $\mathcal{A}_\delta$ as the partition of $\mathbb{R}^2$ induced by $\mathcal{C}$. 
The decision of $\dF(\pi, \sigma+\tau) \leq \delta$ is then uniform among all $\tau \in \mathbb{R}$ in the same face of $\mathcal{A}_\delta$ (for a detailed explanation, we refer to \cite[Section~3]{BenAvrahamKS15} or~\cite{BringmannKN19soda}). 
Thus, it suffices to check, for each face $f$ of $\arr_\delta$, an arbitrarily chosen translation $\tau_f \in f$. Specifically, the Fréchet distance under translation is bounded by $\delta$ if and only if there is some face $f$ of $\arr_\delta$ such that $\dF(\pi, \sigma + \tau_f) \le \delta$. 
Since the arrangement $\arr_\delta$ has size $\Oh(n^4)$ and can be constructed in time $\Oh(n^4)$~\cite{JiangXZ08}, using the standard $\Oh(n^2)$-time algorithm for the fixed-translation Fréchet distance~\cite{EiterM94, AltG95} to decide $\dF(\pi, \sigma + \tau_f) \le \delta$ for each face $f$, we immediately arrive at an $\Oh(n^6)$-time algorithm.

Subsequent improvements~\cite{BenAvrahamKS15,BringmannKN19soda} speed up the decision of $\dF(\pi, \sigma + \tau_f)\le \delta$ for all faces $f$ by choosing an appropriate ordering of the translations $\tau_f$ and designing data structures that avoid recomputing some information for ``similar'' translations, leading to an $\Oh(n^{4.667})$-time algorithm. Still, these works rely on computing the arrangement $\arr_\delta$ of worst-case size $\Theta(n^4)$, and a conditional lower bound indeed rules out $\Oh(n^{4-\epsilon})$-time algorithms~\cite{BringmannKN19soda}.

\subparagraph*{Drawback: The arrangement size bottleneck.}
Despite the worst-case arrangement size of~$\Theta(n^4)$ and the conditional lower bound in~\cite{BringmannKN19soda}, which indeed constructs such large arrangements, one might hope that realistic instances often have much smaller arrangements. If so, a combination with a practical implementation of the fixed-translation Fréchet distance could already give an algorithm with reasonable running time.
Unfortunately, this is not the case: our experiments in this paper exhibit typical arrangement sizes between $10^6$ to $10^8$ for curves of length $n \approx 200$, see Figure \ref{fig:decider-bbcalls} in Section~\ref{sec:experiments}. Also see Figure~\ref{fig:arrangement} which illustrates a large arrangement already on curves with 15 vertices, subsampled from our benchmark sets of realistic curves.

This renders a purely arrangement-based approach infeasible:
As existing implementations for the Fréchet distance typically answer queries within few microseconds, we would expect an average decision time between a few seconds and several minutes already for a single decision query for the Fréchet distance under translation. Thus, a reasonable approximation of the distance value via binary search would take between a minute and over an hour.

\begin{figure}
\begin{center}
	\includegraphics[width=.4\textwidth]{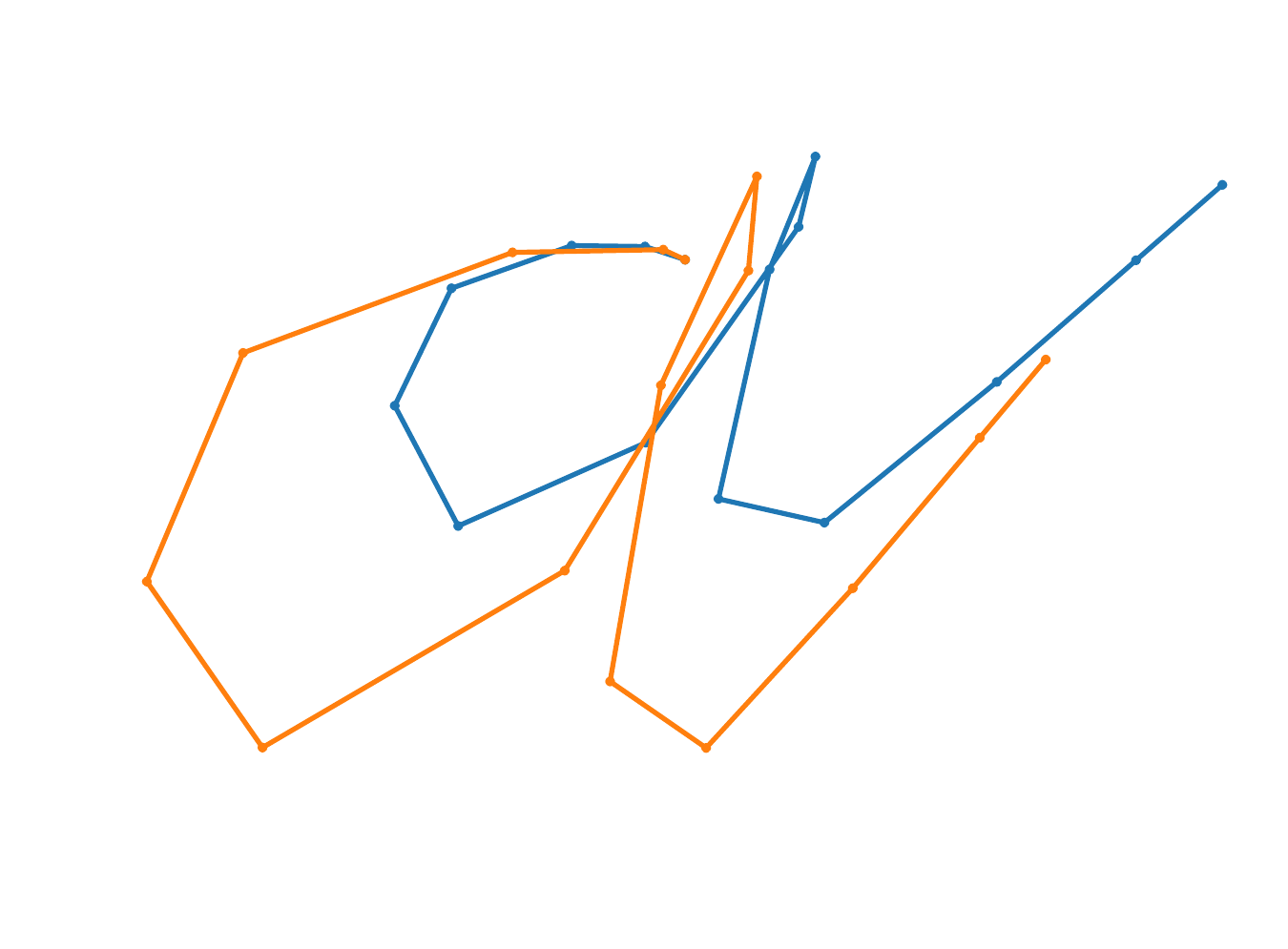}
	\hspace{1cm}
	\includegraphics[width=.4\textwidth]{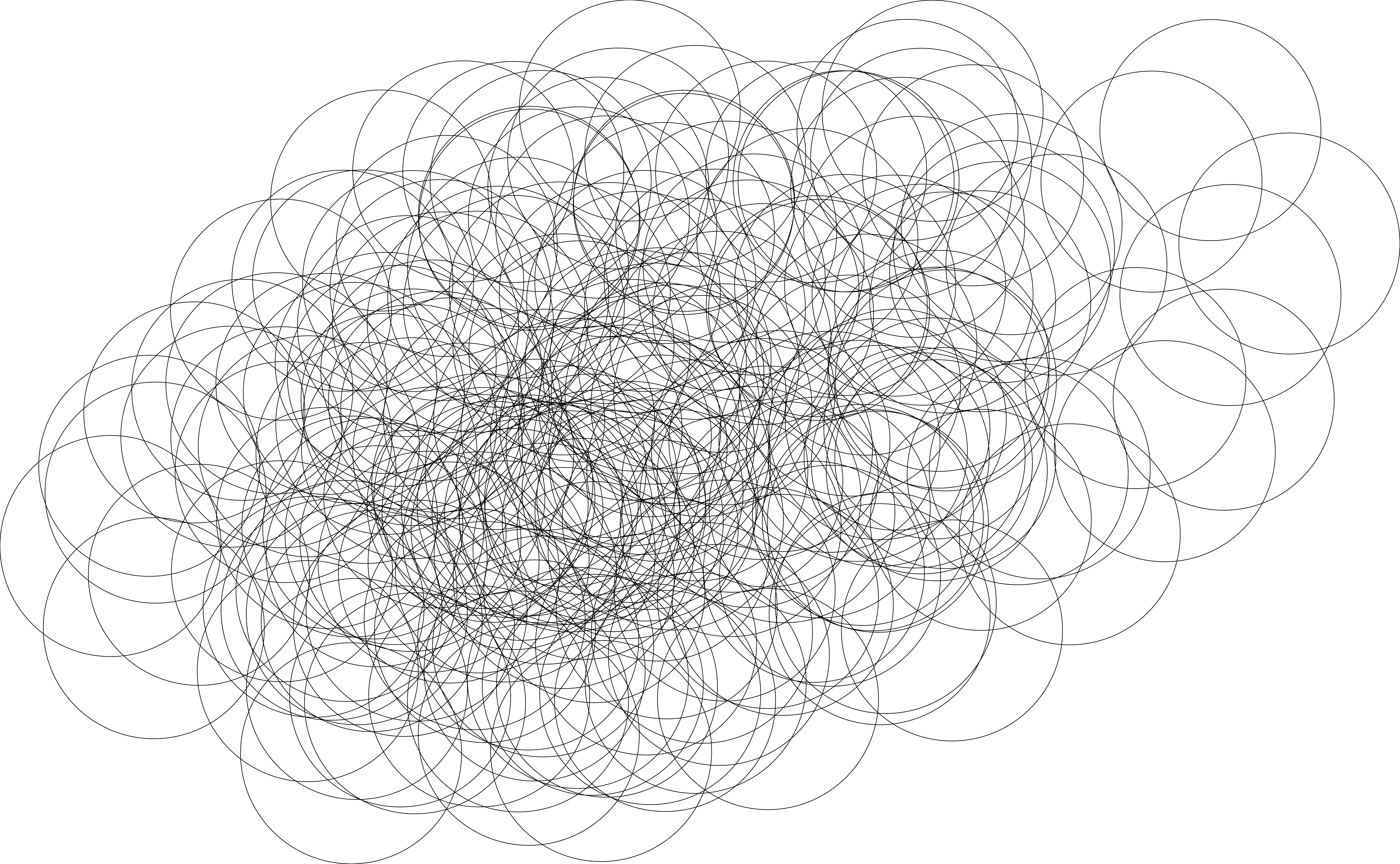}
\end{center}
\caption{Example curves $\pi,\sigma$ (left) together with their arrangement $\arr_\delta$ (right), $\delta = \dtransF(\pi, \sigma)$.}
\label{fig:arrangement}
\end{figure}

\subsection{View II: A Global Lipschitz Optimization problem}
\label{sec:Lipschitzview}

A second view on the Fréchet distance under translation results from a simple observation: For any polygonal curves $\pi,\sigma$ and any translation $\tau\in \R^2$, we have $|\dF(\pi, \sigma + \tau)- \dF(\pi, \sigma)| \le \|\tau\|_2$, see Section~\ref{sec:prelim}. As a consequence, the Fréchet distance under translation is the minimum of a function $f(\tau) \coloneqq \dF(\pi, \sigma + \tau)$ that is 1-Lipschitz (i.e., $|f(x)-f(x+y)| \le \|y\|_2$ for all $x,y$). This suggests to study the problem also from the viewpoint of the generic algorithms developed for optimizing Lipschitz functions by the continuous optimization community.

Following the terminology of \cite{HansenJ95chapter}, in an \emph{unconstrained bivariate global Lipschitz optimization problem}, we are given an objective function $f: \R^2 \to \R$ that is 1-Lipschitz, and the aim is to minimize $f(x)$ over $x\in B\coloneqq [a_1,b_1]\times [a_2,b_2]$; we can access $f$ only by evaluating it on (as few as possible) points $x\in B$. Note that in this abstract setting, we cannot optimize $f$ exactly, so we are additionally given an error parameter $\eps>0$ and the precise task is to find a point $x \in B$ such that $f(x) \le \min_{z \in B} f(z) + \eps$. 

Global Lipschitz optimization techniques have been studied from an algorithmic perspective for at least half a century~\cite{Piyavskii72}. This suggest to explore the use of the fast algorithms developed in this context to obtain at least an \emph{approximate} decider for the discrete Fréchet distance under translation. Indeed, our problem fits into the above framework, if we take the following considerations into account:
\begin{enumerate}[(1)]
\item \textbf{Finite Box Domain:} While we seek to minimize $f(\tau) = \dF(\pi, \sigma + \tau)$ over $\tau \in \R^2$, the above formulation assumes a finite box domain $B$. To reconcile this difference, observe that any translation $\tau$ achieving a Fréchet distance of at most $\delta$ must translate the first (last) point of $\sigma$ such that the first (last) point of $\pi$ is within distance at most $\delta$. Thus, any feasible translation $\tau$ must be contained in the intersection of the two corresponding disks, and we can use any bounding box of this intersection as our box domain $B$.
\item \textbf{(Approximate) Decision Problem:} 
While we seek to decide ``$\min_\tau f(\tau) \le \delta$'', the above formulation solves the corresponding minimization problem. Note that approximate minimization can be used to \emph{approximately} solve the decision problem, but \emph{exactly} solving the decision problem is impossible in the above framework. 
\item \textbf{\boldmath Oracle Access to $f(\tau)$:} Evaluation of $f(\tau)$ corresponds to computing the Fréchet distance of $\pi$ and $\sigma+\tau$, for which we can use previous fast implementations \cite{sigspatial1, sigspatial2,sigspatial3,BringmannKN19socg}. (Actually, these algorithms were designed to answer decision queries of the form ``$f(\tau) \le \delta$?''; we discuss this aspect at the end of this section.)
\end{enumerate}
\begin{figure}
\begin{center}
	\includegraphics[width=.4\textwidth]{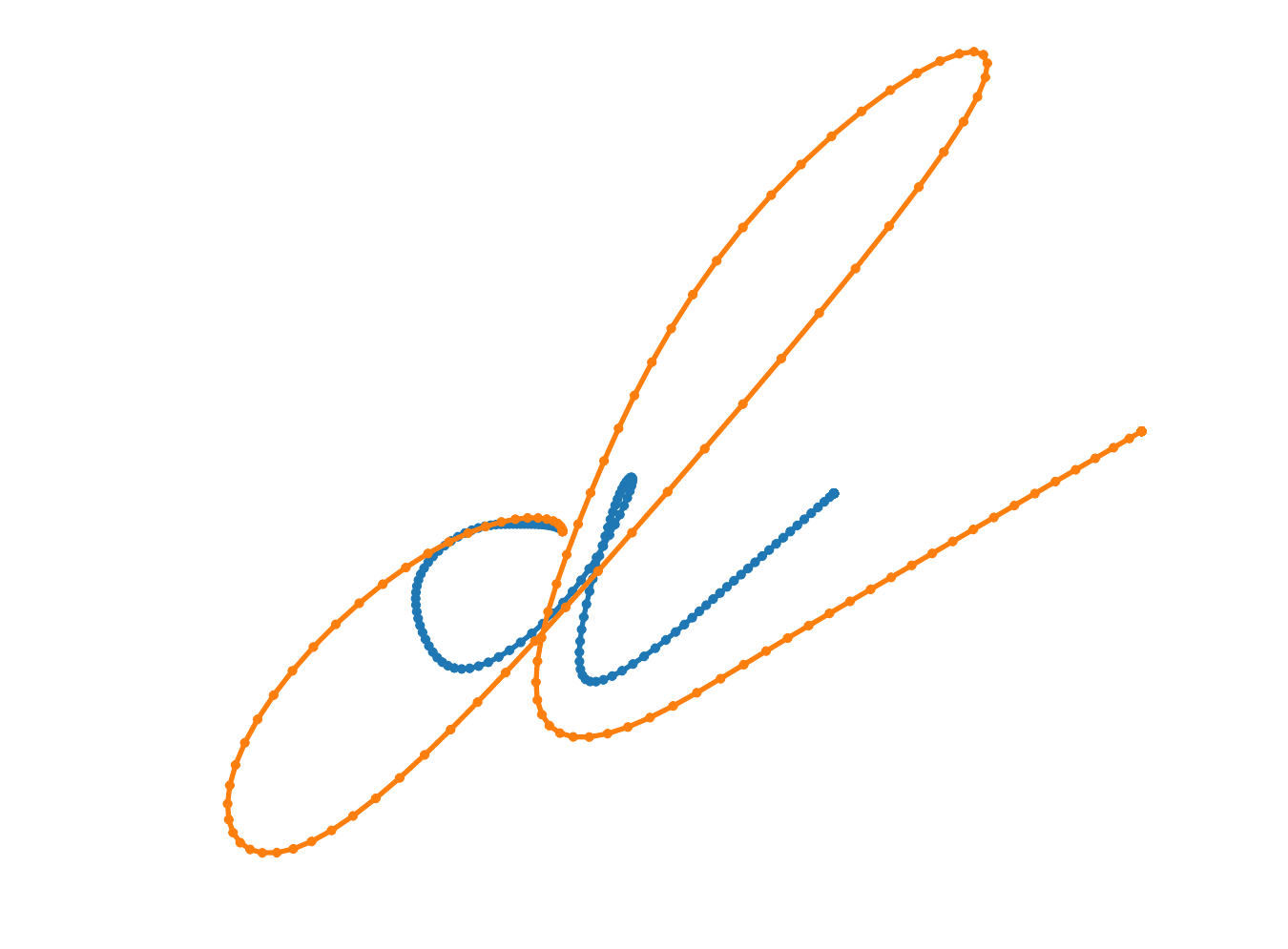}
	\hspace{1cm}
	\includegraphics[width=.4\textwidth]{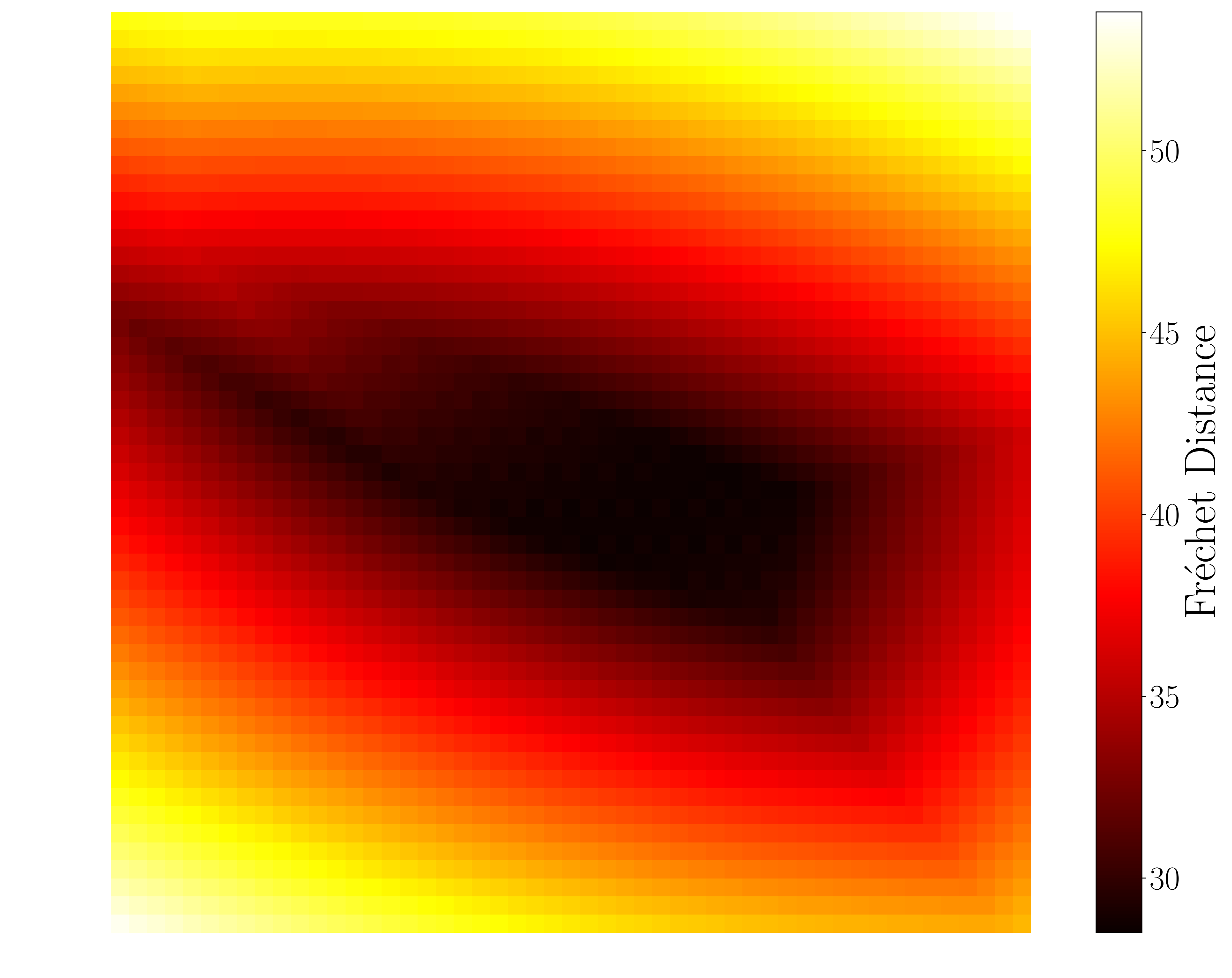}
\end{center}
\caption{Example curves $\pi, \sigma$ (left) together with a plot of the resulting non-convex objective function $f(\tau) = \dF(\pi, \sigma + \tau)$. For a closer look at the area close to the optimal translation (and highly non-convex small-scale artefacts), we refer to Figure~\ref{fig:Lipschitzartefacts}.}
\label{fig:Lipschitzview}
\end{figure}
In Figure~\ref{fig:Lipschitzview}, we illustrate our view of the Fréchet distance under translation as Lipschitz optimization problem. As the figure suggests, on many realistic instances, the problem appears well-behaved (almost convex) at a global scale; using the Lipschitz property, one should be able to quickly narrow down the search space to small regions of the search space\footnote{For an illustration that highly non-convex behavior may still occur at a local level, we refer to Figure~\ref{fig:Lipschitzartefacts}.}. Particularly for this task, it is very natural to consider branch-and-bound approaches, as pioneered by Galperin~\cite{Galperin85, Galperin87, Galperin88, Galperin93} and formalized by Horst and Tuy~\cite{Horst86, HorstT87, HorstT96}, since these have been applied very successfully for low-dimensional Global Lipschitz optimization (and non-convex optimization in general).

\begin{figure}
	\centering
	\includegraphics[width=0.5\textwidth]{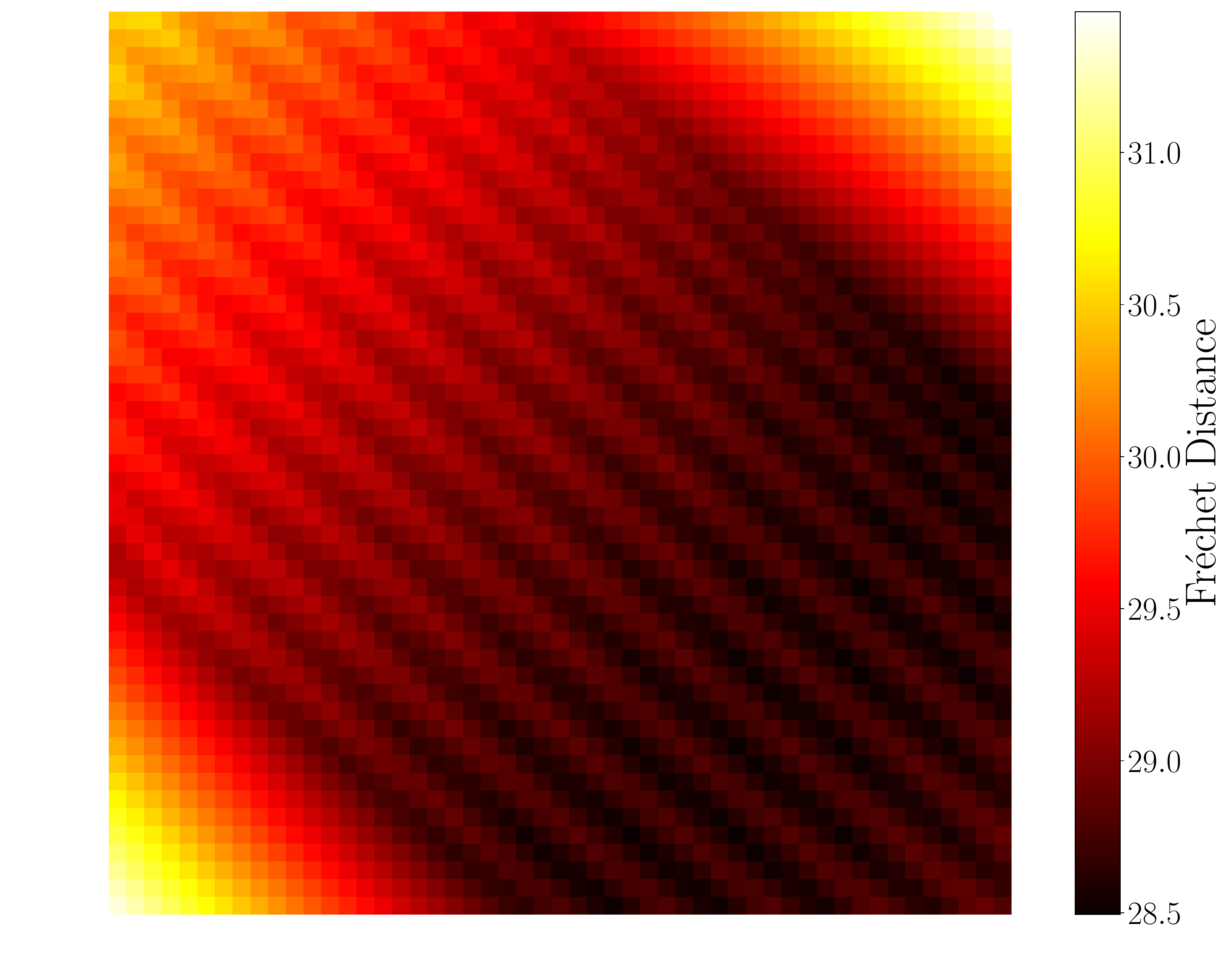}
	\caption{Highly non-convex artefacts of the objective function at a local scale, resulting particularly from the notion of traversals in the \emph{discrete} Fréchet distance.}
	\label{fig:Lipschitzartefacts}
\end{figure}

On a high level, in this approach we maintain a global upper bound $\tilde{\delta}$ and a list of search boxes $B_1, \dots, B_b$ with lower bounds $\ell_1,\dots, \ell_b$ (i.e., $\min_{\tau\in B_i} f(\tau) \ge \ell_i$) obtained via the Lipschitz condition. We iteratively pick some search box $B_i$ and first try to improve the global upper bound $\tilde{\delta}$ or the local lower bound $\ell_i$ using a small number of queries $f(\tau)$ with $\tau \in B_i$ (and exploiting the Lipschitz property). If the local lower bound exceeds the global upper bound, i.e., $\ell_i > \tilde{\delta}$, we drop the search box $B_i$, otherwise, we split $B_i$ into smaller search boxes. The procedure stops as soon as $\tilde{\delta} \le (1+\eps) \min_i \ell_i$, which proves that $\tilde{\delta}$ gives a $(1+\epsilon)$-approximation to the global minimum.

Specifically, we arrive at the following branch-and-bound strategy proposed by Gourdin, Hansen and Jaumard~\cite{GourdinHJ94}. We specify it by giving the rules with which it (i) attempts to update the global upper bound, (ii) selects the next search box from the set of current search boxes, (iii) splits a search box if it remains active after bounding, and (iv) determines the local lower bounds.\footnote{See~\cite{HansenJ95chapter} for a precise formalization of the generic branch-and-bound algorithm that leaves open the instantiation of these rules. In any case, we give a self-contained description of our algorithms in Section~\ref{sec:decider} and~\ref{sec:valueComputation}.}
\begin{enumerate}[(i)]
\item \textbf{Upper Bounding Rule:} We evaluate $f$ at the center $\tau_i$ of the current search box $B_i$. 
\item \textbf{Selection Rule:} We pick the search box with the smallest lower bound (ties are broken arbitrarily).
\item \textbf{Branching Rule:} We split the current search box along its longest edge into $2$ equal-sized subproblems.
\item \textbf{Lower Bounding Rule:} We obtain the local lower bound $\ell_i$ as $f(\tau_i)-d$ where $d$ is the half-diameter of the current box. (Since $f$ is 1-Lipschitz, we indeed have $\min_{\tau\in B_i} f(\tau) \ge \ell_i$.)
\end{enumerate}
One may observe that the chosen selection rule (also known as Best-Node First) is a no-regret strategy in the sense that no other selection rule, \emph{even with prior knowledge of the global optimum}, considers fewer search boxes (see, e.g.,~\cite[Section 7.4]{Wolsey98}).

\subparagraph*{Drawback: Inexactness.}

Unfortunately, the above branch-and-bound approach for Lipschitz optimization fundamentally cannot return an exact global optimum, and thus yields only an approximate decider. 

In a somewhat similar vein, in the above framework we assume that we can evaluate $f(\tau)$ quickly. Previous implementations for the fixed-translation Fréchet distance focus on the decision problem ``$f(\tau)\le \delta$?'', not on determining the value $f(\tau)$. Both precise computations (via parametric search) or approximate computations (using a binary search up to a desired precision) are significantly more costly, raising the question how to make optimal use of the cheaper decision queries. 

%% file: trunk/decider.tex
\section{Contribution I: An Exact Decider by Combining Both Views} \label{sec:contribution} \label{sec:decider}

Our first main contribution is engineering an exact decider for the discrete Fréchet distance under translation by combining the two approaches.
On a high level, we \emph{globally} perform the branch-and-bound strategy described in the Lipschitz optimization view in Section~\ref{sec:Lipschitzview}, but use as a base case a \emph{local} version of the arrangement-based algorithms of Section~\ref{sec:arrView} once the arrangement size in a search box is sufficiently small. As each search box is thus resolved exactly, this yields an exact decider.
More precisely, our final algorithm is a result of the following steps and adaptations:

\begin{enumerate}[(1)]
\item \textbf{Fréchet Decision Oracle.} We adapt the currently fastest implementation of a decider for the continuous fixed-translation Fréchet distance~\cite{BringmannKN19socg} to the discrete fixed-translation Fréchet distance. Furthermore, to handle many queries for the same curve pair \emph{under different translations} quickly, we incorporate an implicit translation so that curves do not need to be explicitly translated for each query translation $\tau$. 
\item \textbf{Objective Function Evaluation.} For our exact decider, the branch-and-bound strategy in Section~\ref{sec:Lipschitzview} simplifies significantly: We do not maintain a global upper bound and local lower bounds $\ell_i$, but for each box only test whether $f(\tau_i) \le \delta$ (if so, we return YES) or whether $f(\tau_i) > \delta + d$ (this corresponds to updating the local lower bound beyond $\delta$, i.e., we may drop the box completely). Therefore, we may use an arbitrary selection rule. Note that we only require decision queries to the fixed-translation Fréchet algorithm.
\item \textbf{Base Case.} We implement a \emph{local} arrangement-based algorithm: For a given search box~$B_i$, we (essentially) construct the arrangement $\arr\cap B_i$ using CGAL~\cite{cgalarrangement}, and test, for each face $f$ of $\arr \cap B_i$, some translation $\tau' \in f$ for $f(\tau')\le \delta$. This yields the algorithm that we may use as a base case.
\item \textbf{Base Case Criterion.} For each search box, we compute an estimate of its arrangement complexity. If this estimate is smaller than a (tunable) parameter $\sizeparameter$, or the depth of the branch-and-bound recursion for the current search box exceeds a parameter $\depthparameter$, then we use the localized arrangement-based algorithm.
\item \textbf{Benchmark and Choice of Parameters.} We choose the size and depth parameters $\sizeparameter, \depthparameter$ guided by a benchmark set that we create from a set of handwritten characters and synthetic GPS trajectories.
\end{enumerate}

\newcommand{\cbox}{\ensuremath{\tau_B}\xspace}
\newcommand{\dbox}{\ensuremath{d_B}\xspace}

\begin{algorithm}[t]
\begin{algorithmic}[1]
\Procedure{Decider}{$\pi, \sigma, \delta$}
\State decide trivial NO instances with empty initial search box quickly
\State $Q \gets \Call{Fifo}{\text{initial search box}}$ \label{l:fifo}
	\While{$\mathrm{Q} \neq \emptyset$}
	\State $B$ $\gets$ extract front of search box queue $Q$
	\If{$\Call{FréchetDistance}{\pi, \sigma + \cbox} > \delta + \dbox/2$} \label{l:positive_decider} \Comment{\emph{Lower Bounding}}
		\State skip $B$
	\EndIf
	\If{$\Call{FréchetDistance}{\pi, \sigma + \cbox} \leq \delta$} \label{l:negative_decider} \Comment{\emph{Upper Bounding}}
		\State \Return YES
	\EndIf
	\State
	\State $u \gets \text{upper bound on arrangement size inside $B$}$ \label{l:intersect_decider}  
	\If{$u = 0$} \label{l:monotonebox_decider} \Comment{\emph{Arrangement-based Base Case}}
		\State skip $B$
		\ElsIf{$u \leq \sizeparameter \text{ \textbf{or} } \text{layer of $B$ is $\depthparameter$}$} \label{l:arrif_decider}
		\If{local arrangement-based algorithm on $\pi, \sigma, \delta, B$ returns YES}
			\State \Return YES
		\Else
			\State skip $B$
		\EndIf
	\EndIf \label{l:end_arr_decider}
	\State 
	\State halve $B$ along longest edge and push resulting child boxes to $Q$\label{l:branch_decider} \Comment{\emph{Branching}}
\EndWhile
\State \Return NO
\EndProcedure
\end{algorithmic}
	\caption{Algorithm for deciding the Fréchet distance under translation. We use \cbox to denote the center of the box $B$ and \dbox to denote the length of the diagonal of $B$.}
\label{alg:decider}
\end{algorithm}

The pseudocode of the resulting algorithm is shown in Algorithm \ref{alg:decider}.
In the remainder of this section, we describe the details of our Fréchet-under-translation decider.
We first describe the details of the \emph{local} arrangement-based algorithm which serves as the base case for our decider.


\subsection{Local Arrangement-Based Algorithm} \label{sec:arrangement}

Recall that given two polygonal curves $\pi, \sigma$ and a decision distance $\delta$, the set of circles of the arrangement $\mathcal{A}_\delta$ is
\[
	\mathcal{C} \coloneqq \{ C_\delta(\pi_i - \sigma_j) \mid \pi_i \in \pi,\; \sigma_j \in \sigma \},
\]
where $C_r(p)$ denotes the circle of radius $r \in \mathbb{R}$ around $p \in \mathbb{R}^2$. The arrangement is then defined as the partition of $\mathbb{R}^2$ induced by $\mathcal{C}$. In particular, the decision of $\dF(\pi, \sigma+\tau) \leq \delta$ is uniform for each $\tau \in \mathbb{R}$ in the same face of $\mathcal{A}_\delta$ (for a detailed explanation, we refer to \cite[Section 3]{BenAvrahamKS15} or \cite{BringmannKN19soda}). Thus, as already described in Section \ref{sec:arrView}, it suffices to evaluate a representative translation from each face of the arrangement by running a fixed-translation Fréchet decider query on it to reach a Fréchet under translation query decision.

For integration into our branch-and-bound approach where each node in the branch-and-bound tree corresponds to a search box $B$, the base case task is to decide whether there is some $\tau\in B$ with $\dF(\pi, \sigma+\tau) \le \delta$. For this task, we consider \emph{local} arrangements, i.e., arrangements restricted to $B$. A circle $C\in \mathcal{C}$ contributes to the local arrangement of $B$ if the boundary of $C$ intersects the box. In other words, $C$ is relevant for the arrangement of $B$ if $C$ either is completely contained in $B$ or $C$ intersects the boundary of $B$. In particular, $C$ does not contribute to the local arrangement if it contains $B$ completely.

\subparagraph*{Estimation of local arrangement sizes.}
Given a search box $B$, a simple way to estimate the size of the local arrangement for $B$, i.e., the arrangement restricted to $B$, is to consider the number of circles in $\mathcal{C}$ that contribute to it. We can obtain this number naively, by iterating over all $|\mathcal{C}| \le nm$ circles of the global arrangement and check if they contribute to the local arrangement (by checking for intersection and containment as described above). Let this number be denoted by $c$. The maximal number of nodes in the arrangement is then bounded by $u \coloneqq 2(c + c^2)$, as this is the maximal number of intersections between two circles and a circle and the box. In particular, if $u = 0$, then the arrangement in the box belongs to a single face and all translations in $B$ are equivalent for our decision question.

As a simple optimization, we may stop counting contributing circles once our estimate exceeds the threshold $\sizeparameter$. A more sophisticated optimization builds a geometric data structure (specifically a kd-tree) to quickly retrieve all contributing circles without checking all circles in $\mathcal{C}$ naively. We discuss this approach in Section~\ref{sec:valueComputation}, as the expensive preprocessing for constructing this data structure only amortizes in the value computation setting.

\subparagraph*{Construction of local arrangement.}
For a search box $B$ with an estimate smaller than $\sizeparameter$, we construct an arrangement $\arr_B$. To this end, we adapt our arrangement size estimation to also return the set $\mathcal{C}_B$ of circles intersecting $B$ or being contained in $B$. 
Note that computing topologically correct geometric arrangements on such a circle set is a challenging task, as it requires the usage of arbitrary precision numbers to reliably test for intersections and orderings of those intersections. Thus, we use the state-of-the-art computational geometry library CGAL~\cite{cgalarrangement} to build our circle arrangements.\footnote{Specifically, we use the exact predicates and exact computation kernels as this is necessary for CGAL arrangements. The significantly faster kernel for inexact computation is not suitable for the CGAL arrangement package (although, surprisingly, for most instances it actually worked). Being able to use a faster kernel for arrangements should significantly improve our implementation's performance.}
Unfortunately, CGAL only provides methods for building a global arrangement and not an arrangement restricted to a bounding box, thus we always build the whole arrangement of the circles in $ \mathcal{C}_B$ instead of just the arrangement restricted to the box $B$. Alternatively, we could indeed compute circular arcs restricted to the bounding box and then build the arrangement of those arcs. However, due to the rather expensive construction of these arcs, this seems wasteful compared to a direct computation. Thus, a practical performance improvement of our approach could be achieved by directly computing an arrangement with a box restriction. Furthermore, we use the standard bulk-insertion interface for building the arrangement.

\subparagraph*{Resulting local arrangement-based algorithm.}
Finally, given the arrangement $\arr_B$ of the circles $\mathcal{C}_B$, we may simply test a translation $\tau$ for each face $f$ of $\arr_B$ that intersects $B$. In fact, for efficiency, we do this by testing each vertex $\tau$ of $\arr_B$ (even for vertices outside of $B$, as due to the expensive construction of $\arr_B$, it pays off to make the rather cheap tests for positive witnesses also outside of $B$); observe that this ensures that each face $f$ is indeed tested. We return YES if and only if some vertex $\tau$ of $\arr_B$ achieves $\dF(\pi, \sigma + \tau)\le \delta$.

\subsection{Decision Algorithm}


Now, we describe our decider (whose pseudocode is given in Algorithm~\ref{alg:decider}) in more detail.
Recall that an exact decider, given curves $\pi=(\pi_1,\dots,\pi_n), \sigma=(\sigma_1,\dots,\sigma_m)$ and a distance $\delta$, decides whether the Fréchet distance under translation of $\pi$ and $\sigma$ is at most $\delta$, i.e., whether $\dtransF(\pi, \sigma) \leq \delta$.

\subparagraph*{Preprocessing.}
As a first step, we aim to determine an initial search box. Since any $\tau \in \R^2$ with $\dF(\pi, \sigma + \tau) \le \delta$ implies that $\norm{\pi_1 - (\sigma_1+\tau)},\norm{\pi_n - (\sigma_m+\tau)}\le \delta$, we must have that $\tau$ is in the intersection $I \coloneqq D_\delta(\pi_1-\sigma_1) \cap D_\delta(\pi_n-\sigma_m)$, where $D_r(p)$ denotes the disk of radius $r$ around $p$. If this intersection is empty, i.e., $\pi_1-\sigma_1$ and $\pi_n-\sigma_m$ have a distance more than $2\delta$, we return NO immediately. Otherwise, we take a bounding box of the intersection.\footnote{In fact, we use a slightly more refined search box by incorporating additionally the extreme points of both curves.}

\subparagraph{Branch-and-Bound.}
We implement the recursive branch-and-bound strategy using a FIFO queue $Q$ of search boxes (corresponding to a breadth-first search) that is initialized with the initial search box. As long as there are undecided boxes in the queue, we take the first such box $B$ and try to resolve it using the upper bounding rule (point (i) in View II) and the lower bounding rule (point (iv) in View II), which are both derived by queries to the fixed-translation Fréchet distance decider using the center point \cbox{} of the box as translation. Specifically, if $\dF(\pi, \sigma + \cbox) \leq \delta$ (line \ref{l:positive_decider} in Algorithm \ref{alg:decider}), we have found a witness translation and can return YES. The lower bounding rule (line \ref{l:negative_decider}) tests if $\dF(\pi, \sigma + \cbox) > \delta + \dbox/2$, i.e., if the distance at the center point is larger than the test distance $\delta$ plus the maximal distance of any point in the box to the center \cbox{}, i.e., the half-diagonal length $\dbox/2$. If so, by the Lipschitz property, we know that the any translation in $B$ yields a Fréchet distance larger than $\delta$ and thus we can drop $B$.

If neither rule applies, we check our termination criterion of the branch-and-bound strategy. To this end, in line \ref{l:intersect_decider}, we calculate a good upper bound $u$ on the size of the local arrangement for $B$ as described in Section~\ref{sec:arrangement}. If $u = 0$, the arrangement for $B$ consists of a single face, i.e., each translation $\tau \in B$ is equivalent for our decision problem, and we can skip the box since we have already tested the translation $\cbox\in B$. Otherwise, in line \ref{l:arrif_decider}, if $u \neq 0$, we check if the number is bounded by a size parameter $\sizeparameter$ or the depth of the current search box (in the implicit recursion tree) is bounded by a depth parameter $\depthparameter$. If so, we run the local arrangement-based algorithm to decide $B$.

If none of the above rules decide the search box $B$, we split it along its longer side into two equal-sized child boxes and push them to the queue. If all boxes have been dropped without finding a witness translation, we have verified that any translation $\tau \in B$ yields $\dF(\pi, \sigma + \tau) > \delta$ and may safely return NO.



\subparagraph*{Low-level optimizations.}
For further practical speed-ups, we employ several low-level optimizations, which we briefly mention here (for further details, we refer to the source code of our implementation).

For each box in the branch-and-bound tree we need a differently translated curve. However, often we barely access the nodes of the translated curve. For example, if already the start nodes of the curves are too far, we do not need to consider the remainder. Thus, it seems wasteful to translate each point of the curves before calling the fixed-translation Fréchet decider. To avoid this overhead, we lazily translate the necessary parts of a curve on access. In fact, while the currently fastest implementation of the fixed-translation Fréchet distance decider~\cite{BringmannKN19socg} uses a preprocessing of the curves that computes all prefix lengths and extrema of the curves, we only need to perform this preprocessing once, as all computed information is either invariant under translations (for the prefix lengths) or can just be shifted by the translation (for the extrema).

Furthermore, while the initial bounding box is derived from the discs around the translation between the start nodes and the translation between the end nodes, later child boxes in the branch-and-bound tree might violate this condition. We therefore re-check this condition on creating child boxes.
Additionally, in line \ref{l:arrif_decider} of Algorithm \ref{alg:decider} we check if the depth parameter $\depthparameter$ is reached. This can actually already be done before line \ref{l:negative_decider}, which we also do in the implementation, but for the sake of brevity, we present it differently in the pseudocode.

%% file: trunk/valComputation.tex
\section{Contribution II: Approximate Computation of the Distance Value}
\label{sec:introValueComputation}
\label{sec:valueComputation}

In this section we present our second main contribution: an algorithm for computing the value of the Fréchet distance under translation.
Thus, we now focus on the functional task of computing the value $\dtransF(\pi,\sigma) = \min_{\tau \in \R^2} \dF(\pi, \sigma + \tau)$, in contrast to the previously discussed decision problem ``$\dtransF(\pi,\sigma) \le \delta$?''.
In theory, one could use the paradigm of parametric search~\cite{Megiddo83}, see~\cite{BenAvrahamKS15, BringmannKN19soda} for details for the discrete case. However, it is rarely used in practice as it is non-trivial to code, and computationally costly.
Instead, as in most conceivable settings an estimate with small multiplicative error $(1 \pm \eps)$ with, e.g., $\eps = 10^{-7}$, suffices, we consider the problem of computing an estimate in $(1\pm \epsilon)\dtransF(\pi,\sigma)$.

\begin{algorithm}[t]
\begin{algorithmic}[1]
\Procedure{\lmf}{$\pi, \sigma$}
\State Preprocessing: build data structures for fast arrangement estimation and construction\label{l:kd_tree}
\State compute initial distance interval $[\deltaLB, \deltaUB]$ containing $\dtransF(\pi, \sigma)$ \label{l:initial_estimates}
\State initialize global upper bound $\tilde{\delta} \gets \deltaUB$
\State $Q \gets \Call{PriorityQueue}{\text{initial search box $B_{1}$ with local lower bound $\ell_{B_1}\gets \deltaLB$}}$ \label{l:pq}
\While{$Q \neq \emptyset$}
	\State $B \gets$ box with smallest local lower bound $\ell_B$ in $Q$
	\If{$\tilde{\delta} \leq \ell_B(1 + \epsilon)$} \label{l:bound_start}
		\State skip $B$
	\EndIf
	\If{$\Call{FréchetDistance}{\pi, \sigma + \cbox} \le \tilde{\delta}$} \label{l:ub_improvement1} \Comment{\emph{Upper/Lower Bounding}}
	\State compute value $\dF(\pi, \sigma+\cbox)$ with high precision and update $\tilde{\delta}$ and $\ell_B$ \label{l:value_comp1}
	\Else
		\If{$\Call{FréchetDistance}{\pi, \sigma + \cbox} > \tilde{\delta} + \dbox/2$} 
			\State skip $B$
		\EndIf
		\State compute value $\dF(\pi, \sigma + \cbox)$ with coarse precision and update $\ell_B$ \label{l:value_comp2}
	\EndIf

	\If{$\tilde{\delta} \leq \ell_B(1 + \epsilon)$}
		\State skip $B$
	\EndIf \label{l:bound_end}

	\State $u \gets \text{upper bound on arrangement size inside $B$ for $\delta \in [\ell_B, \tilde{\delta}]$}$ \label{l:arr_start}
	\If{$u = 0$} \Comment{\emph{Arrangement-based Base Case}}
	  \State skip $B$
	\ElsIf{$u \leq \sizeparameter$ or layer of $B$ is $\depthparameter$} \label{l:arrif_lmf}
		\State update $\tilde{\delta}$ via binary search over arrangement algorithm on $B$ and $\delta \in [\ell_B, \tilde{\delta}]$ \label{l:ub_improvement2}
		\State skip $B$
	\EndIf \label{l:arr_end}
	\State 
	\State push child boxes of $B$ to $Q$ with local lower bounds set to $\ell_B$ \label{l:branch_lmf} \Comment{\emph{Branching}}
\EndWhile
\State \Return $\tilde{\delta}$
\EndProcedure
\end{algorithmic}
\caption{Algorithm of our Lipschitz-Meets-Fréchet (LMF) algorithm for approximate value computation. We use \cbox to denote the center of the box and \dbox to denote the length of the diagonal.}
\label{alg:lmf}
\end{algorithm}

There are several possible approaches to obtain an approximation with multiplicative error $(1\pm\epsilon)$ for arbitrarily small $\epsilon > 0$:
\begin{enumerate}
\item \textbf{$\eps$-approximate Set:} A natural approach underlying previous approximation algorithms~\cite{alt2001matching} is to generate a set of $f(1/\eps)$ candidate translations $T$ such that the best translation $\tau$ among this set gives a $(1+\eps)$-approximation for the Fréchet distance under translation. 
Specifically, it is simple to obtain a bounding box $B$ of side length $\Oh(\delta)$ for the optimal translation $\tauopt$ (see the 2-approximation in Section~\ref{sec:prelim} together with the preprocessing described in Section~\ref{sec:decider}). We impose a grid of side length at most $(\eps/\sqrt{2}) \delta $ so that each each point in $B$ is within distance $\eps \delta$ of some grid point. Since the Fréchet distance is Lipschitz, this yields a $(1+\epsilon)$-approximate set. Unfortunately, this set is of size $\Theta(1/\eps^2)$ which is prohibitively large for approximation guarantees such as $\eps = 10^{-7}$. \\
    \textit{Remark:} In the context of global Lipschitz optimization, this approach is known as the \emph{passive algorithm} whose performance generally is dominated by (the adaptive) branch-and-bound methods.

\item \textbf{Binary Search via Decision Problem:} A further canonical approach is to reduce the $(1+\eps)$-approximate computation task to the decision problem using a binary search. 
Formally, let $\deltaopt$ denote the Fréchet distance under translation. Starting from a simple $2$-approximation $\deltaUB$ (see Section~\ref{sec:prelim}, or, more precisely, the initial estimates discussed later in this section), we use a binary search in the interval $[0.5 \cdot \deltaUB,\deltaUB]$, terminating as soon as we arrive at an interval of length $[a,b]$ with $b \le (1+\epsilon)a$. As this takes only $\Oh(\log(1/\eps))$ iterations to obtain an $(1+\eps)$-approximation, this approach is much more suitable to obtain a desired guarantee of $\eps = 10^{-7}$.

\item \textbf{Lipschitz-only Optimization:} The main drawback of the generic Lipschitz optimization algorithms discussed in Section~\ref{sec:Lipschitzview} was that they cannot be used to derive an exact answer. This drawback no longer applies for approximate value computation. We can thus use a pure branch-and-bound algorithm for global Lipschitz optimization. In particular, we will use the same strategy as our fastest solution, however, we never use the arrangement-based algorithm, but only terminate at a search box once the local lower bound and global upper bound provide a $(1+\eps)$-approximation.

\item \textbf{Our solution, Lipschitz-meets-Fréchet:} We follow our approach of combining Lipschitz optimization with arrangement-based algorithms (described in Section~\ref{sec:approach}) to compute a $(1+\eps)$-approximation of the distance value. As opposed to the decision algorithm, we indeed maintain a global upper bound $\tilde{\delta}$ and local lower bounds $\ell_i$ for each search box $B_i$. To update these bounds, we approximately evaluate the objective function $f(\tau)$ using a tuned binary search\footnote{We tune the binary search by distinguishing the precision with which we want to evaluate $f(\tau)$; intuitively, it pays off to evaluate $f(\tau)$ with high precision if this evaluation yields a better global upper bound, while for improvements of a local lower bound, a cheaper evaluation with coarser precision suffices.} over the fixed-translation Fréchet decider algorithm. We stop branching in a search box $B_i$ if either the global upper bound $\tilde{\delta}$ is at most $\ell_i(1 + \eps)$, or a base case criterion similar to the decision setting applies. As selection strategy, we employ the no-regret strategy of choosing the box with the smallest lower bound first. The base case performs a binary search using the local arrangement-based \emph{decision} algorithm; thus, our upper bound on the arrangement size must hold for \emph{all} $\delta$ in the search interval.
The pseudocode of our solution is shown in Algorithm \ref{alg:lmf}.

We present the details of our approach in the remainder of this section.
As our experiments reveal, our solution generally outperforms the above described alternatives, see Section~\ref{sec:experiments}.

\subparagraph*{Remark:}
To enable a fair comparison of the Lipschitz-meets-Fréchet (LMF)  approach to the alternative approaches of Binary Search and Lipschitz-only optimization, we take care that the low-level optimizations for LMF described in the reminder of this section are also applied to these approaches, as far as applicable. In particular, we use the same method to obtain initial estimates for the desired value for LMF, Binary Search and Lipschitz-only optimization, and adapt the kd-tree-based data structure used to speed-up estimation and construction of arrangements for LMF also for Binary Search (note that these tasks do not apply to Lipschitz-only optimization).



\end{enumerate}

We now present details of our solution for the (approximate) value computation setting, the \lmf algorithm.
We first consider the base case (which differs from the base case of the decider, given in Section~\ref{sec:arrangement}), before we discuss further details.

\subsection{Local Arrangement-Based Algorithm for Value Computation}

Our base case problem is the following: Given curves $\pi, \sigma$, a \emph{test distance interval} $I=[\deltaLB,\deltaUB]$ and a search box $B$, we let $\deltaopt \coloneqq \min_{\tau \in B} \dF(\pi, \sigma+\tau)$ and ask to determine whether $\deltaopt \in I$, and if so, an estimate $\delta'$ with $|\delta' - \deltaopt| \le \eps$.

The central idea is to solve this task via a binary search for $\deltaopt \in I$ using our local arrangement-based algorithm of Section~\ref{sec:arrangement} to decide queries of the form ``$\deltaopt \le \delta$?'' for any given $\delta$. For this algorithm to run quickly, we need that for \emph{any} queried distance $\delta$, the corresponding local arrangement for the test distance $\delta$ is small. To this end, we seek to obtain a strong upper bound for the local arrangement size over \emph{worst-case} $\delta \in I$.

\subparagraph{Estimation of local arrangement sizes.}
Given an interval $I=[\deltaLB, \deltaUB]$ of test distances, instead of the \emph{circles} defined in Section~\ref{sec:arrangement}, we consider the set of \emph{annuli}
\[
	\mathcal{D} \coloneqq \{ D_{\deltaUB}(\pi_i - \sigma_j) \setminus D_{\deltaLB}(\pi_i - \sigma_j) \mid \pi_i \in \pi,\; \sigma_j \in \sigma \},
\]
where $D_r(p)$ denotes the disk of radius $r \in \mathbb{R}$ around $p \in \mathbb{R}^2$. Clearly, if a circle $C_{\delta}(\pi_i - \sigma_j)$ contributes to the local arrangement  of $B$ for \emph{some} test distance $\delta\in [\deltaLB, \deltaUB]$, then the corresponding annulus $D_{\deltaUB}(\pi_i - \sigma_j) \setminus D_{\deltaLB}(\pi-\sigma_j)$ intersects $B$ or is contained in $B$. Thus by determining the number $d$ of annuli $a\in \mathcal{D}$ that intersect $B$ or are contained in $B$, we may bound the local arrangement size for $B$ for any $\delta \in [\deltaLB, \deltaUB]$ by $u\coloneqq 2(d+d^2)$ (analogously to Section~\ref{sec:arrangement}).

To obtain the above upper bound efficiently, we implement a geometric search data structure based on the \emph{kd-tree}. Specifically, we build a kd-tree on the set of center points of all annuli in $\mathcal{D}$. Given a search box $B$, we seek to determine all centers of annuli $a\in \mathcal{D}$ that intersect $B$ or contain $B$. While this condition can be described using a constant (but large) set of simple primitives, evaluating this test frequently for many kd-tree nodes is costly. Thus, to determine whether a node in the kd-tree needs to be explored, we use a more permissive, but cheaper test which essentially approximates the search box $B$ by its center point: we search for all candidate points that are contained in an annulus of width roughly $|I|$ plus half the diameter of $B$, centered at the center of $B$, and test for each such point whether the corresponding annulus in $\mathcal{D}$ indeed intersects $B$. 

Again, we implement this search for contributing annuli such that we return the centers of all found annuli. This can subsequently be used by the local arrangement-based algorithm to quickly construct the arrangement for each query. Furthermore, we again stop the search as soon as the numbers of such annuli exceeds $\sizeparameter$.

\subparagraph*{Binary Search via local arrangement-based algorithm.}

To obtain the desired estimate for $\deltaopt$ in the case that our size estimate is bounded by $\sizeparameter$, we use a binary search via our local arrangement-based algorithm. As a low-level optimization to speed-up the construction of the local arrangement for a query distance $\delta$, we pass the centers of contributing annuli to the local arrangement-based algorithm. Furthermore, as described in Section~\ref{sec:arrangement}, we let the arrangement-based decision algorithm test \emph{all} vertices in the arrangement of all circles $\mathcal{C}_B$ contributing to the search box $B$, not only vertices in $B$. As this can only decrease the returned estimate (by finding a corresponding witness), this does not affect correctness of the algorithm. 

\subsection{Overview and Details for LMF}

The pseudocode of the LMF algorithm is shown in Algorithm \ref{alg:lmf}. When referring to lines in the remainder of this section, we refer to lines in this algorithm. Before we address some aspects and optimizations in detail, we give a short overview over the algorithm. First, note that as our selection strategy is different from the decider setting, we now use a priority queue for the boxes, see line \ref{l:pq}. In lines \ref{l:bound_start} to \ref{l:bound_end} the bounding happens and in lines \ref{l:arr_start} to \ref{l:arr_end} we check if the base case criterion applies, and if it does, determine the value for this box using the arrangement-based approach.
Finally, in line \ref{l:branch_lmf} we branch if we did not already skip the box.

\subparagraph*{Initial estimates.}
In line \ref{l:initial_estimates} we calculate initial estimates for the upper and lower bound. To this end, we consider the translation $\taustart$ (resp. $\tauend$) that aligns the first (resp. last) points of $\pi,\sigma$ as it yields a $2$-approximation $\deltastart \coloneqq \dF(\pi,\sigma+\taustart)$ (resp. $\deltaend \coloneqq \dF(\pi, \sigma + \tauend)$). Using the best of both approximations, our initial estimation interval for $\dtransF(\pi, \sigma)$ is $[\deltaLB, \deltaUB] \coloneqq [\max\{\deltastart,\deltaend\}/2, \min\{\deltastart,\deltaend\}]$, see Section~\ref{sec:prelim}.

\subparagraph*{Priority queue.}
To implement our smallest-lower-bound-first selection rule, we use a priority queue to organize the search boxes, using the local lower bounds as keys. Recall that this yields a no-regret selection strategy for our branch-and-bound framework.

\subparagraph*{Objective function evaluation: Computing Fréchet distance via Fréchet decider.}
To update our global upper bound and local lower bounds, we need to determine Fréchet distance \emph{values} rather than decisions (which were sufficient for our decider), see lines \ref{l:value_comp1} and \ref{l:value_comp2}. However, we do not always need a very precise calculation. While the upper bound is global and thus an improvement might lead to significant progress by dropping a number of search boxes, the lower bound only has an effect on the box itself and on its children. Thus, we use a coarse distance computation (i.e., an approximation up to a larger additive constant) for the lower bound in line \ref{l:value_comp2}, but a more precise calculation for the upper bound in line \ref{l:value_comp1}.

In two cases (lines \ref{l:ub_improvement1} and \ref{l:ub_improvement2}) we are only interested in the exact Fréchet distance value if it is smaller than the current global upper bound. Thus, as is hidden in the pseudocode, we first check if there is an improvement at all, and only if this is the case, we compute the actual value using a binary search.

\subparagraph*{Additive vs. multiplicative approximation.} Due to rounding issues that occur at decisions depending on extremely small value differences when using fixed precision arithmetic, we use an additive approximation of $\epsilon = 10^{-7}$ instead of a multiplicative approximation to ensure that these issues do not arise on usage of our implementation with arbitrary data sets. Note that all computed distances in our benchmarks have a value larger than $1$, and thus also in terms of multiplicative approximation $(1+\epsilon')$, we have $\epsilon' \le 10^{-7}$.




%

%% file: trunk/experiments.tex
\newcommand{\deltaoptLB}{\delta_{\mathrm{LB}}}
\newcommand{\deltaoptUB}{\delta_{\mathrm{UB}}}

\section{Experiments} \label{sec:experiments}


To engineer and evaluate our approach, we provide a benchmark on the basis of the curve datasets that were used to evaluate the currently fastest fixed-translation Fréchet decider implementation in~\cite{BringmannKN19socg}. In particular, this curve set involves a set of handwritten characters (\characters,~\cite{characters_dataset}) and the data set of the GIS Cup 2017 (\sigspatial, \cite{sigspatial_dataset}). Table~\ref{tab:datasets} gives statistics of these datasets.

\begin{table}[ht]
\centering
\caption{Information about the data sets used in the benchmarks.}
\begin{tabular}{llccc}
\toprule
Data set & Type & \#Curves & Mean \#vertices \\
\midrule
\sigspatial \cite{sigspatial_dataset} & synthetic GPS-like& 20199 & 247.8 \\
\midrule
\multirow{2}{*}{\characters \cite{characters_dataset}} & \multirow{2}{*}{20 handwritten characters} & 2858 & \multirow{2}{*}{120.9} \\
  & & (142.9 per character) \\
\bottomrule
\end{tabular}
\label{tab:datasets}
\end{table}

The aim of our evaluations is to investigate the following main questions:
\begin{enumerate}
\item Is our solution able to decide queries on realistic curve sets in an amount of time that is practically feasible, even when the size of the arrangement suggests infeasibility? 
\item Is our combination of Lipschitz optimization and arrangement-based algorithms for value computation superior to the alternative approaches described in Section~\ref{sec:introValueComputation}?
\end{enumerate}
Furthermore, we aim to provide an understanding of the performance of our novel algorithms.

\subparagraph*{Decider experiments.}

\begin{figure}
\includegraphics[width=\textwidth]{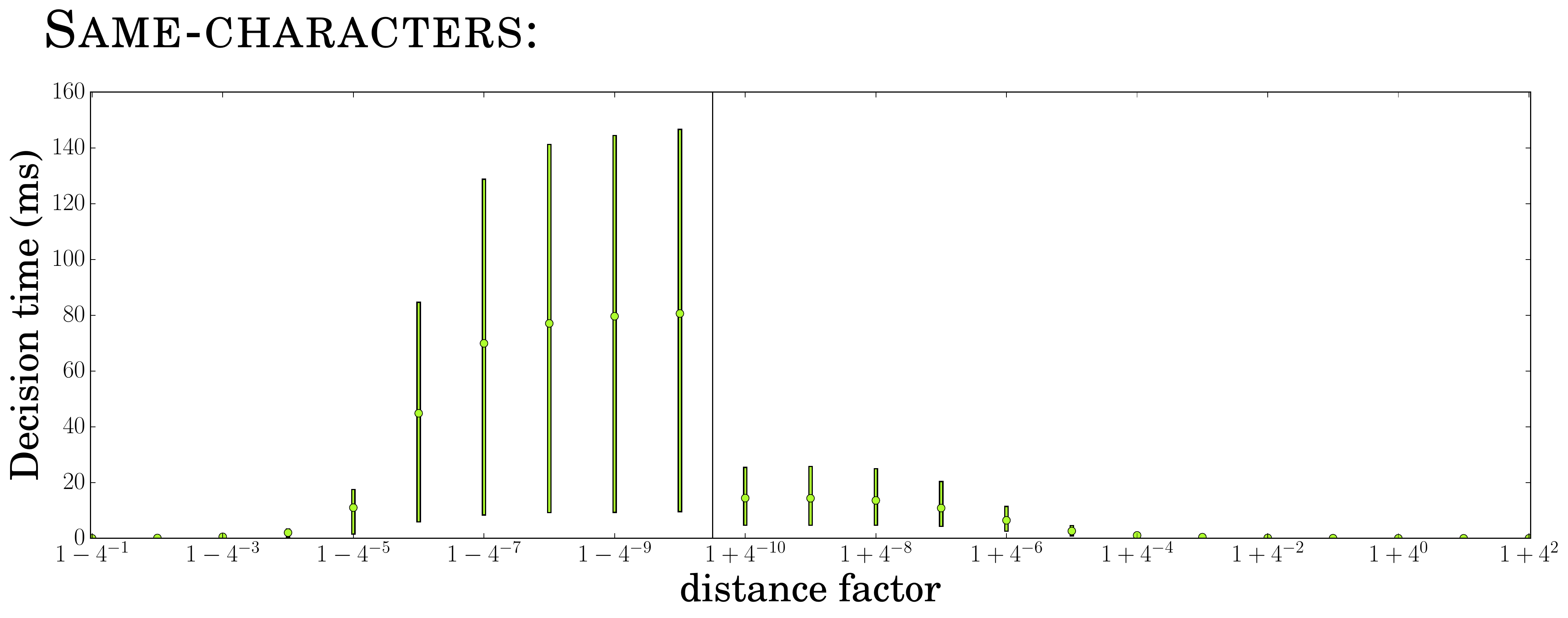}
\includegraphics[width=\textwidth]{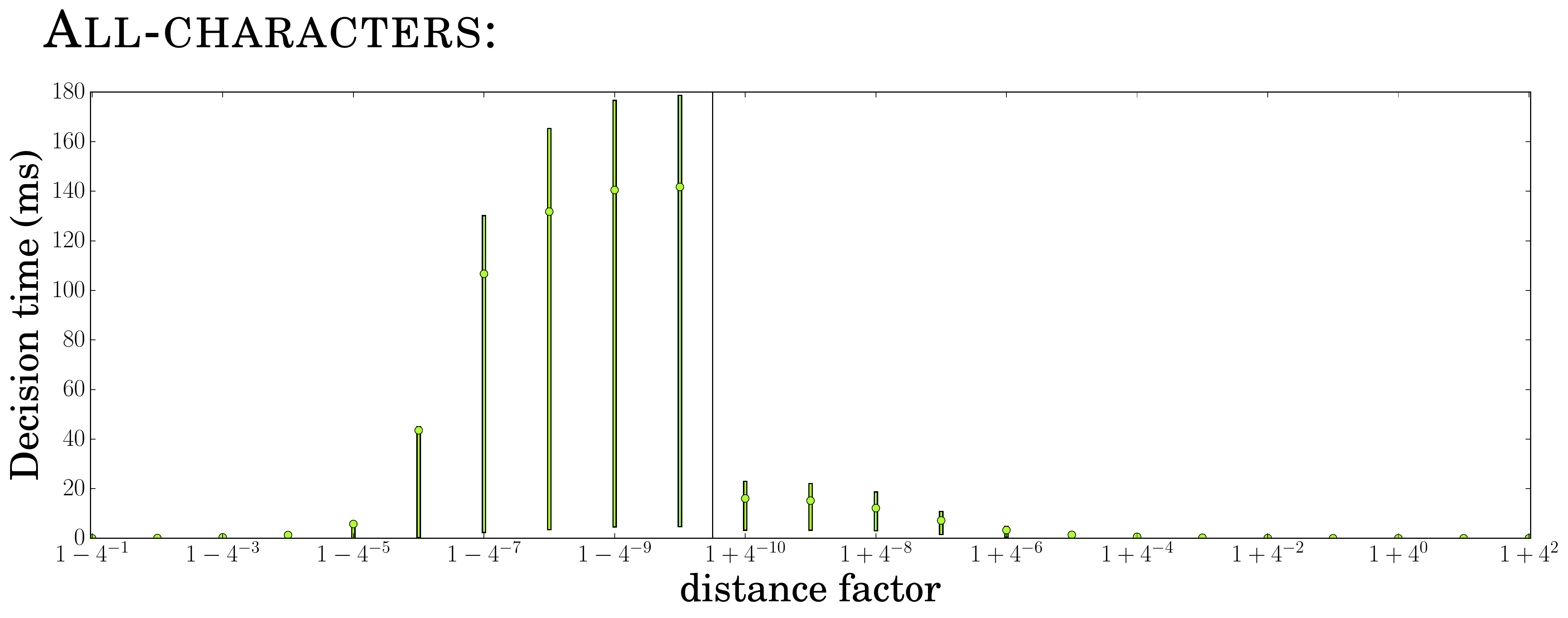}
\includegraphics[width=\textwidth]{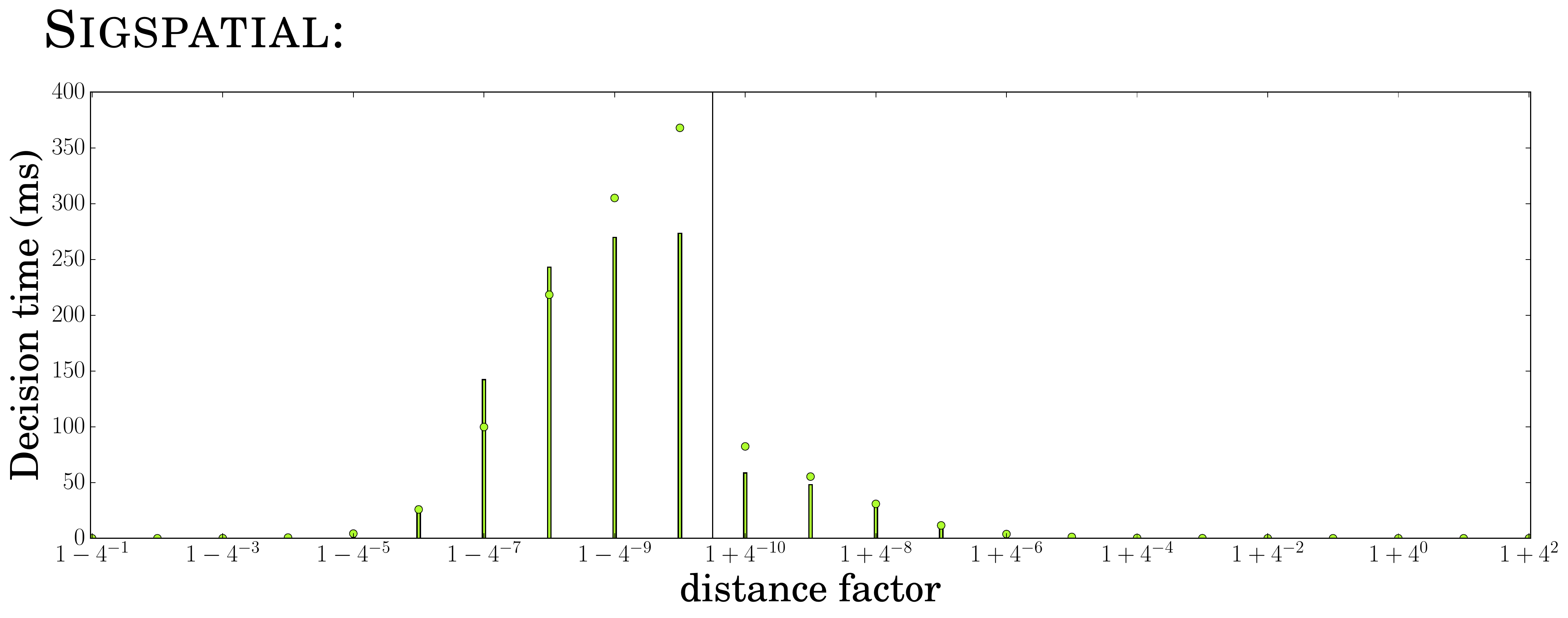}
\caption{Running time for our decider. We plot the mean running times over 1000 NO (or YES) queries with a test distance of approximately $(1-4^{-\ell})$ (or $(1+ 4^{-\ell})$) times the true Fréchet distance under translation, as well as the interval between the lower and upper quartile over the queries.}
\label{fig:decider-times}
\end{figure}

For decision queries of the form ``$\dtransF(\pi, \sigma) \le \delta$?'', we generate a benchmark query set that distinguishes between how close the test distance is to the actual distance of the curves: Given a set of curves $C$, we sample $1000$ curve pairs $\pi, \sigma\in C$ uniformly at random. Using our implementation, we determine an interval $[\deltaoptLB, \deltaoptUB]$ such that $\deltaoptUB - \deltaoptLB \le 2\cdot 10^{-7}$ and $\dtransF(\pi, \sigma) \in [\deltaoptLB,\deltaoptUB]$. For each $\ell \in \{-10, \dots, -1\}$, we add ``$\dtransF(\pi, \sigma) \le (1-4^{\ell})\deltaoptLB$?'' to the query set $C^{\textup{NO}}_\ell$, which contains only NO instances. Similarly, for each $\ell \in \{-10,\dots, 2\}$ we add ``$\dtransF(\pi, \sigma) \le (1+4^{\ell})\deltaoptUB$?'' to the query set $C^{\textup{YES}}_\ell$, which contains only YES instances. We evaluate our decider on this benchmark created for the \characters and \sigspatial data sets.
Furthermore, we give results for a further benchmark set generated from the \characters curve set by sampling, for each of the 20 characters~$c$ included in \characters, 50 curve pairs $\pi, \sigma$ representing the \emph{same character} $c$. This yields a benchmark that has the same size of 1000 query curve pairs, but compares only same-character curves. We show the mean running times on these three benchmark sets in Figure~\ref{fig:decider-times}. As before, we also depict the number of black-box calls of our decider and, as a baseline, an estimate of the arrangement size (and thus the number of black-box calls of a naive approach) in Figure~\ref{fig:decider-bbcalls}. Note that for small ranges of the test distance $\delta$, it may happen that we decide a NO instance without a single black-box call by determining that the distance between $\pi_1 - \sigma_1$ and $\pi_n - \sigma_m$ is larger than $2\delta$; corresponding values below $1$ call are not depicted in Figure~\ref{fig:decider-bbcalls}.

To give an insight for the speed-up achieved over the baseline arrangement-based algorithm that makes a black-box call to the fixed-translation Fréchet decider for each face of the arrangement $\arr_\delta$, in Figure~\ref{fig:decider-bbcalls} we depict both the number of black-box calls to the fixed-translation Fréchet decider made by our implementation, as well as an estimate\footnote{We only give an estimate for the arrangement size, since the size of the arrangement is too large to be evaluated exactly for all our benchmark queries within a day. Specifically, we estimate the number of vertices of the arrangement which closely corresponds to the number of faces by Euler's formula. We give the following estimate: We first determine a search box $B$ for the given decision instance $\pi = (\pi_1,\dots, \pi_n),\sigma=(\sigma_1,\dots,\sigma_m),\delta$ as described for our algorithm. We then sample $S=100000$ tuples $i_1,i_2 \in \{1,\dots, n\},j_1,j_2 \in \{1,\dots, m\}$ and count the number $I$ of intersections of the circles of radius $\delta$ around $\pi_{i_1}-\sigma_{j_1}$ and $\pi_{i_2} -\sigma_{j_2}$ inside $B$. The number $(I/S)\cdot (nm)^2$ is the estimated number of circle-circle intersections in $B$. Adding the number of circle-box intersections, which we can compute exactly, yields our estimate.} for the arrangement size, and thus the number of black-box calls of the baseline approach.


\begin{figure}
\includegraphics[width=\textwidth]{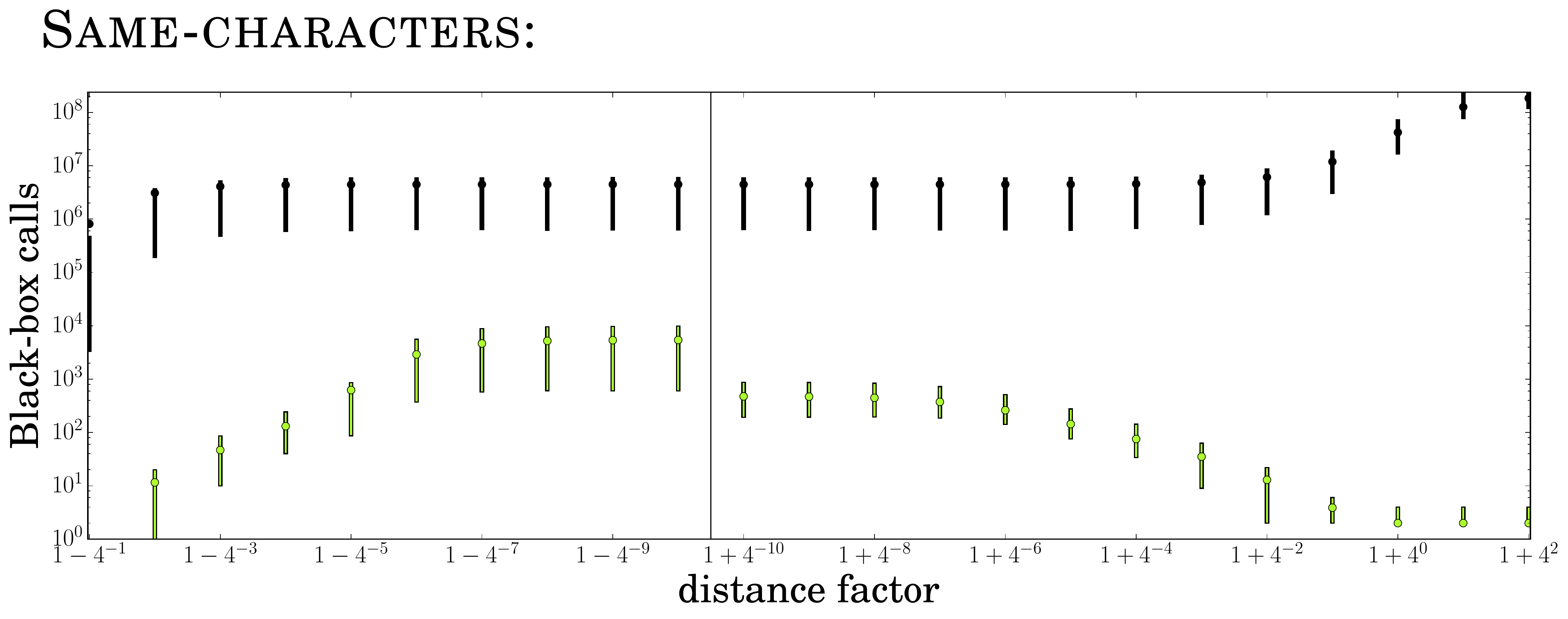}
\includegraphics[width=\textwidth]{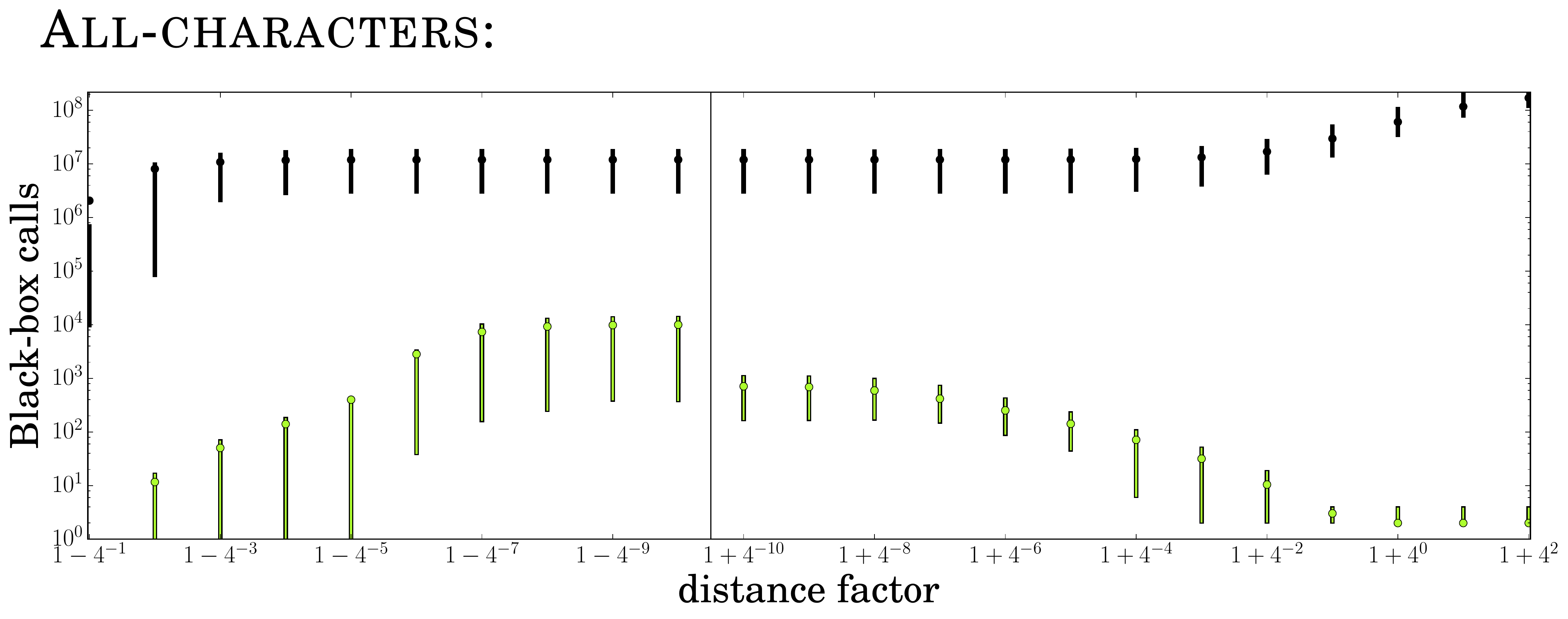}
\includegraphics[width=\textwidth]{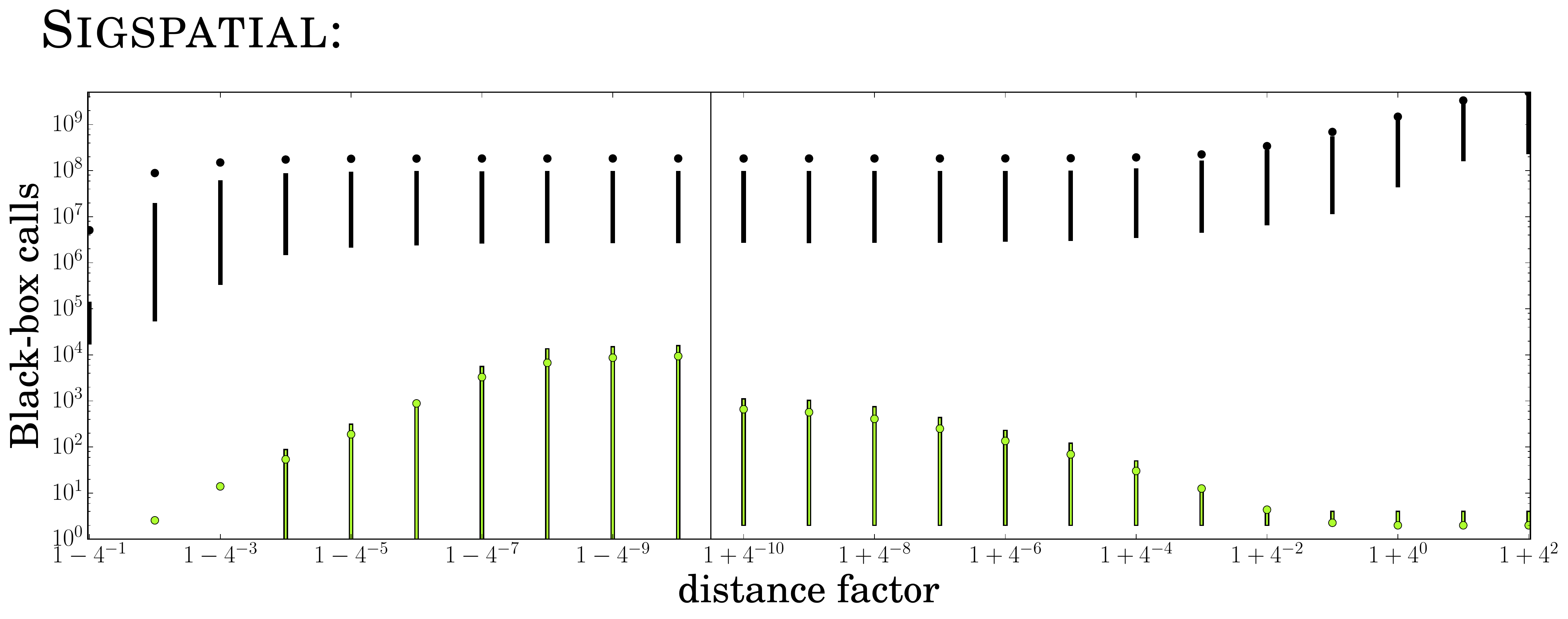}
\caption{Number of black-box calls to the fixed-translation Fréchet decider made by our decider (below, in green), as well as an estimate of the arrangement complexity, i.e., number of calls of a naive algorithm (above, in black).  We plot the mean number of calls and arrangement complexity over 1000 NO (or YES) queries with a test distance of approximately $(1-4^{-\ell})$ (or $(1+ 4^{-\ell})$) times the true Fréchet distance under translation, as well as the interval between the lower and upper quartile over the queries.}
\label{fig:decider-bbcalls}
\end{figure}

We observe that on the above sets, the average decision time ranges from below 1 ms to 142 ms, deciding our \characters benchmark (involving $23,000$ queries) in 628 seconds. Our estimation suggests that a naive implementation of the baseline arrangement-based algorithm would have been worse by more than \emph{three orders of magnitude}, as for each set, the average number of black-box calls to the fixed-translation Fréchet decider is smaller by a factor of at least $1000$ than our estimation of the arrangement size. See Table \ref{table:components_decider} for the detailed timing results of our decider on the benchmarks described above.

\begin{table}
\caption{Time measurements for the components of the decider over the complete decider benchmark sets. In parentheses, we give average values over the total of 23,000 decision instances.}
\begin{tabular}{llrc}
\toprule
\textsc{same-characters} & \multicolumn{2}{c}{\textbf{Time}} &  \textbf{Black-Box Calls}  \\
\midrule
& \multicolumn{2}{c}{429,623 ms} & 26,661,524 \\
& \multicolumn{2}{c}{(18.7 ms per instance)} & (1,159.2 per instance)  \\
\cmidrule(r){2-3} 
& - Preprocessing & 5 ms  \\
\cmidrule(r){2-3}
& - Black-box calls (Lipschitz) & 44,312 ms  \\
\cmidrule(r){2-3}& - Arrangement estimation &157,780 ms  \\
\cmidrule(r){2-3}
& - Arrangement algorithm & 226,469 ms  \\
& \hphantom{bla} * Construction & 148,898 ms \\
& \hphantom{bla} * Black-box calls & 60,156 ms \\
\bottomrule
\toprule
\textsc{all-characters} & \multicolumn{2}{c}{\textbf{Time}} &  \textbf{Black-Box Calls}  \\
\midrule
& \multicolumn{2}{c}{628,043 ms} & 42,781,931 \\
& \multicolumn{2}{c}{(27.3 ms per instance)} & (1,860.08 per instance)  \\
\cmidrule(r){2-3} 
& - Preprocessing & 5 ms  \\
\cmidrule(r){2-3}
& - Black-box calls (Lipschitz) & 50,462 ms  \\
\cmidrule(r){2-3}& - Arrangement estimation &191,177 ms  \\
\cmidrule(r){2-3}
& - Arrangement algorithm & 385,145 ms  \\
& \hphantom{bla} * Construction & 237,043 ms \\
& \hphantom{bla} * Black-box calls & 120,149 ms \\
\bottomrule 
\toprule
\textsc{sigspatial} & \multicolumn{2}{c}{\textbf{Time}} &  \textbf{Black-Box Calls}  \\
\midrule
& \multicolumn{2}{c}{1,207,560 ms} & 31,420,517 \\
& \multicolumn{2}{c}{(52.5 ms per instance)} & (1,366.11 per instance)  \\
\cmidrule(r){2-3} 
& - Preprocessing & 5 ms  \\
\cmidrule(r){2-3}
& - Black-box calls (Lipschitz) & 43,861 ms  \\
\cmidrule(r){2-3}& - Arrangement estimation &913,266 ms  \\
\cmidrule(r){2-3}
& - Arrangement algorithm & 249,268 ms  \\
& \hphantom{bla} * Construction & 155,332 ms \\
& \hphantom{bla} * Black-box calls & 73,934 ms \\
\bottomrule 
\end{tabular}
\label{table:components_decider}
\end{table}

\subparagraph*{Approximate value computation experiments.}

We evaluate our implementation of the algorithm presented in Section \ref{sec:valueComputation} by computing an estimate $\tilde{\delta}$ such that $|\tilde{\delta} - \dtransF(\pi, \sigma)| \le \eps$ with a choice of $\eps = 10^{-7}$.\footnote{Here we use additive rather than multiplicative approximation for technical reasons. Since all computed distances are within $[1.6,120.7]$, this also yields a multiplicative $(1+\eps)$-approximation with $\eps \le 10^{-7}$.} In particular, we compare the performances of the different approaches discussed in Section~\ref{sec:valueComputation}:
\begin{itemize}
\item \textbf{Binary Search:} Binary search using our Fréchet-under-translation decider of Section~\ref{sec:decider}.
\item \textbf{Lipschitz-only:} Algorithm \ref{alg:lmf} without the arrangement, i.e., without lines \ref{l:arr_start} to \ref{l:arr_end}.
\item \textbf{Lipschitz-meets-Fréchet (LMF):} Our implementation as detailed in Section~\ref{sec:valueComputation}.
\end{itemize}
Since simple estimates show that the $\eps$-approximate sets are clearly too costly for $\eps = 10^{-7}$, we drop this approach from all further consideration. We took care to implement all approaches with a similar effort of low-level optimizations.

For our evaluation, we focus on the \characters data set which allows us to distinguish the rough shape of the curves: We subdivide the curve set into the subsets $C_{\alpha}$ for $\alpha \in \Sigma$ (where $\Sigma$ is the set of $20$ characters occurring in \characters). In particular for each character pair $\alpha, \beta \in \Sigma$, we create a sample of $\Nsamples$ curve pairs $(\pi, \sigma)$ chosen uniformly at random from $C_\alpha \times C_\beta$.
For $\Nsamples = 5$, computing the value (up to $\eps = 10^{-7}$) for all $\Nsamples \cdot ({|\Sigma| \choose 2}+|\Sigma|) = 1050$ sampled curve pairs gives the statistics shown in Table \ref{table:totalvalCompBenchmarkN5}.

\begin{table}
\centering
\caption{Statistics for approximate value computation for $\Nsamples = 5$. In parentheses we show the mean values averaged over a total of 1050 instances.}
\begin{tabular}[b]{l ccr}
\textbf{Approach} & \multicolumn{2}{c}{\textbf{Time}} &  \textbf{Black-Box Calls}  \\
\midrule
LMF & \multicolumn{2}{r}{148,032 ms} & 13,323,232 \\
& \multicolumn{2}{r}{(141.0 ms per instance)} & (12,688.8 per instance)  \\
\midrule 
Binary Search & \multicolumn{2}{r}{536,853 ms} & 45,909,628\\
& \multicolumn{2}{r}{(511.3 ms per instance)} &  (43,723.5 per instance) \\
\midrule
Lipschitz-only & \multicolumn{2}{r}{4,204,521 ms} & 820,468,224\\
& \multicolumn{2}{r}{(4,004.3 ms per instance)} &  (781,398.3 per instance) \\
\end{tabular}
\label{table:totalvalCompBenchmarkN5}
\end{table}

Since already for this example the Lipschitz-only approach is dominated by almost a factor of 30 by LMF (and by a factor of almost 8 by binary search), we perform more detailed analyses with $\Nsamples = 100$ only for LMF and binary search. The overall performance is given in Table~\ref{table:totalvalCompBenchmarkN100}. Also here LMF is more than 3 times faster than binary search. To give more insights into the relationship of their running times, we give a scatter plot of the running times of LMF and binary search on the same instances over the complete benchmark in Figure~\ref{fig:binary_lmf_scatter}, showing that binary search generally outperforms LMF only on instances which are comparably easy for LMF as well. The advantage of LMF becomes particularly clear on hard instances.

\begin{figure}
	\centering
	\includegraphics[width=0.8\textwidth]{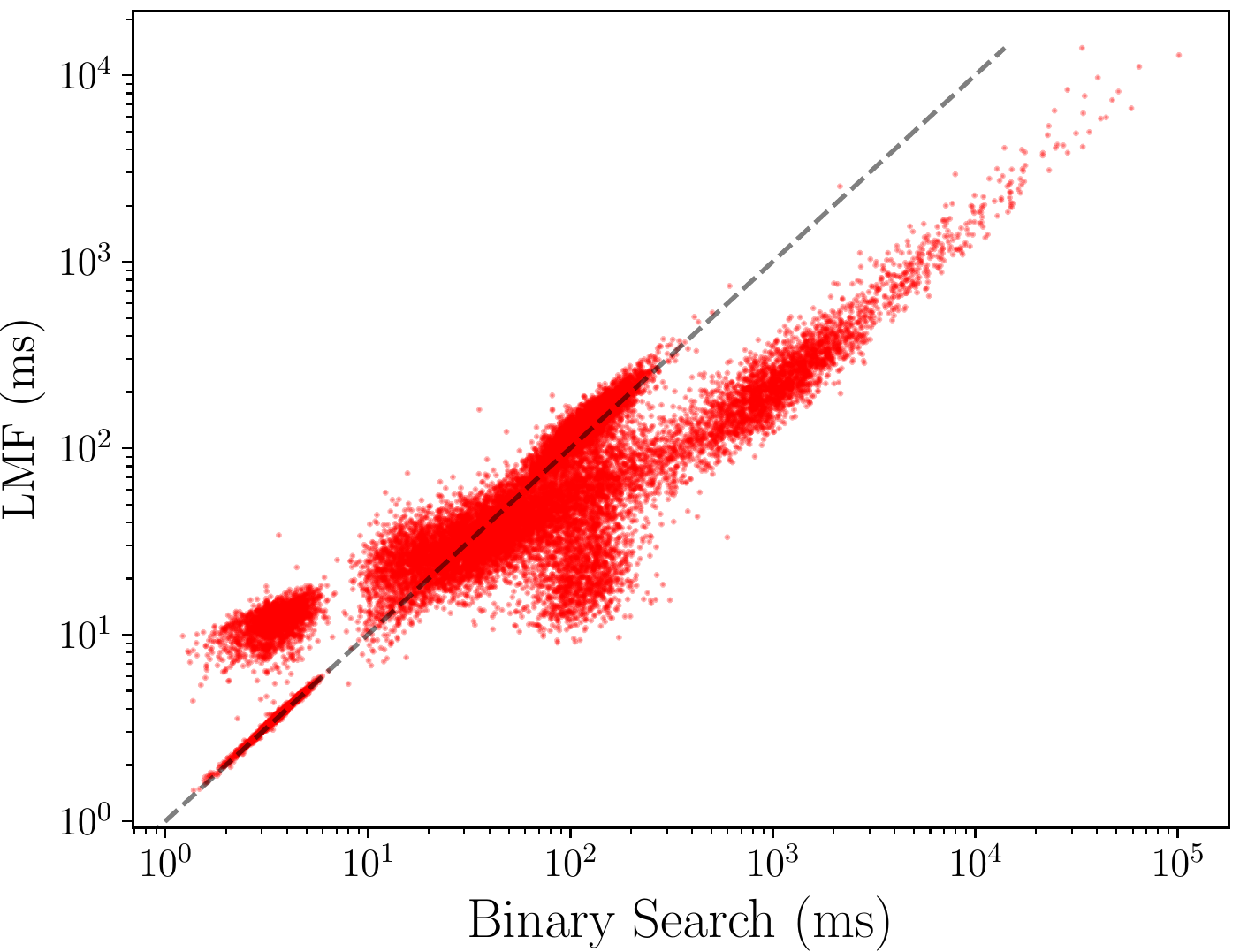}
	\caption{Running times of LMF and binary search on set of randomly sampled \characters curves.}
\label{fig:binary_lmf_scatter}
\end{figure}

Apart from these general statistics for our value computation benchmarks, we depict individual mean computation times and mean number of black-box calls (over all $\Nsamples$ samples) for each character pair $\alpha, \beta \in \Sigma$ in Figures~\ref{fig:valcomp_timesheatmap} and~\ref{fig:valcomp_bbcallsheatmap}. 

\begin{figure}
\includegraphics[width=0.5\textwidth]{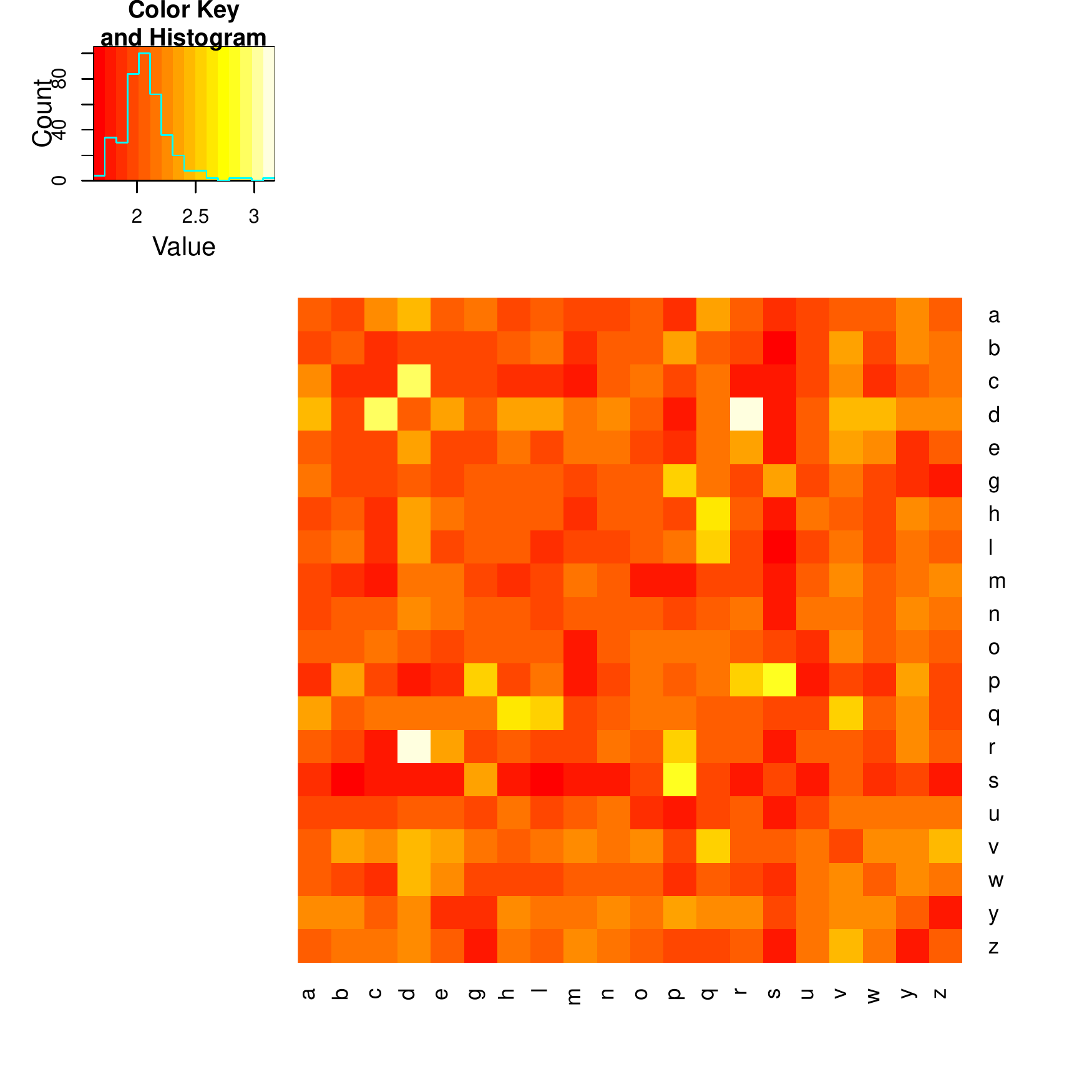}
\includegraphics[width=0.5\textwidth]{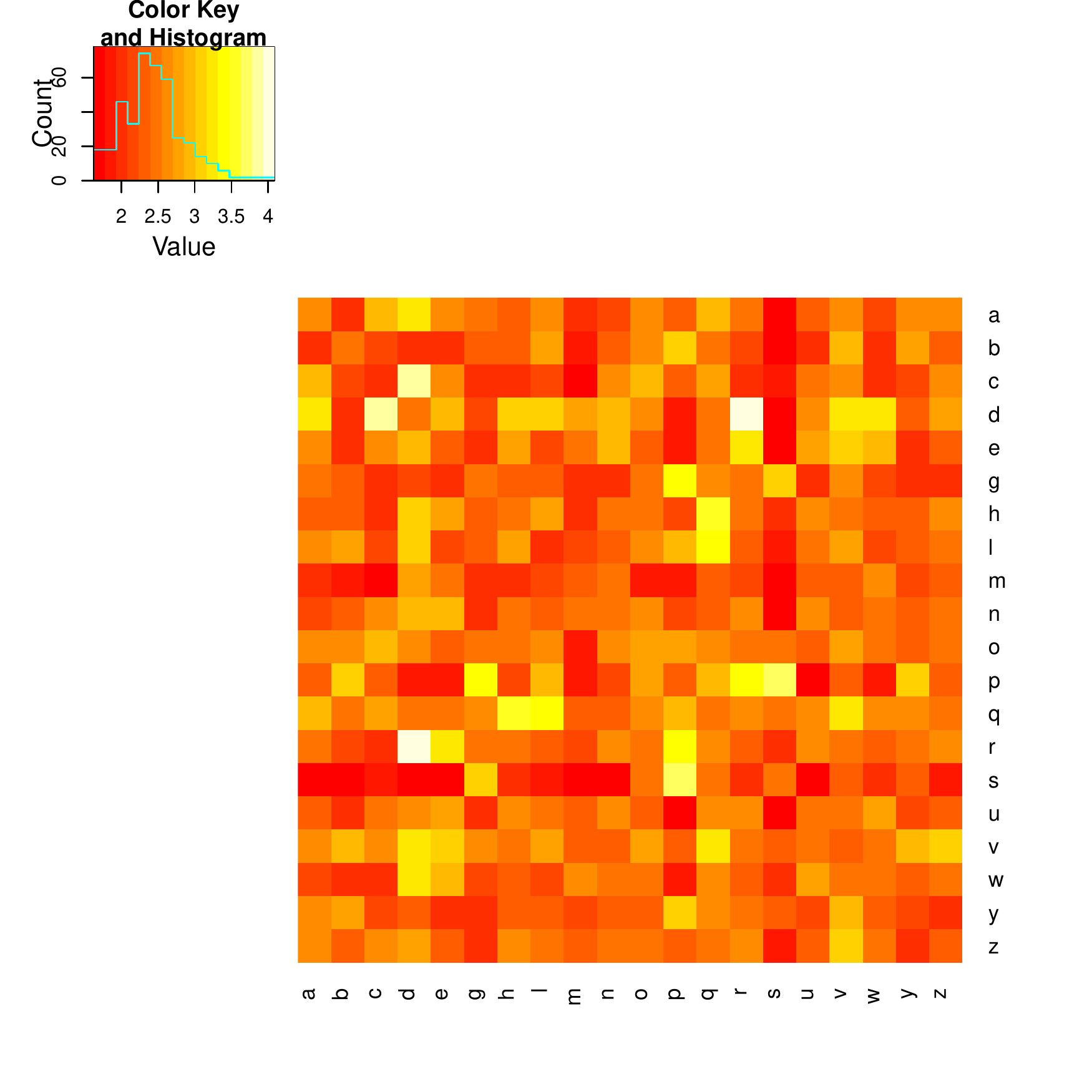}
\caption{Log of mean value computation time in ms for LMF (left) and Binary Search (right).}
\label{fig:valcomp_timesheatmap}
\end{figure}

\begin{figure}
\includegraphics[width=0.5\textwidth]{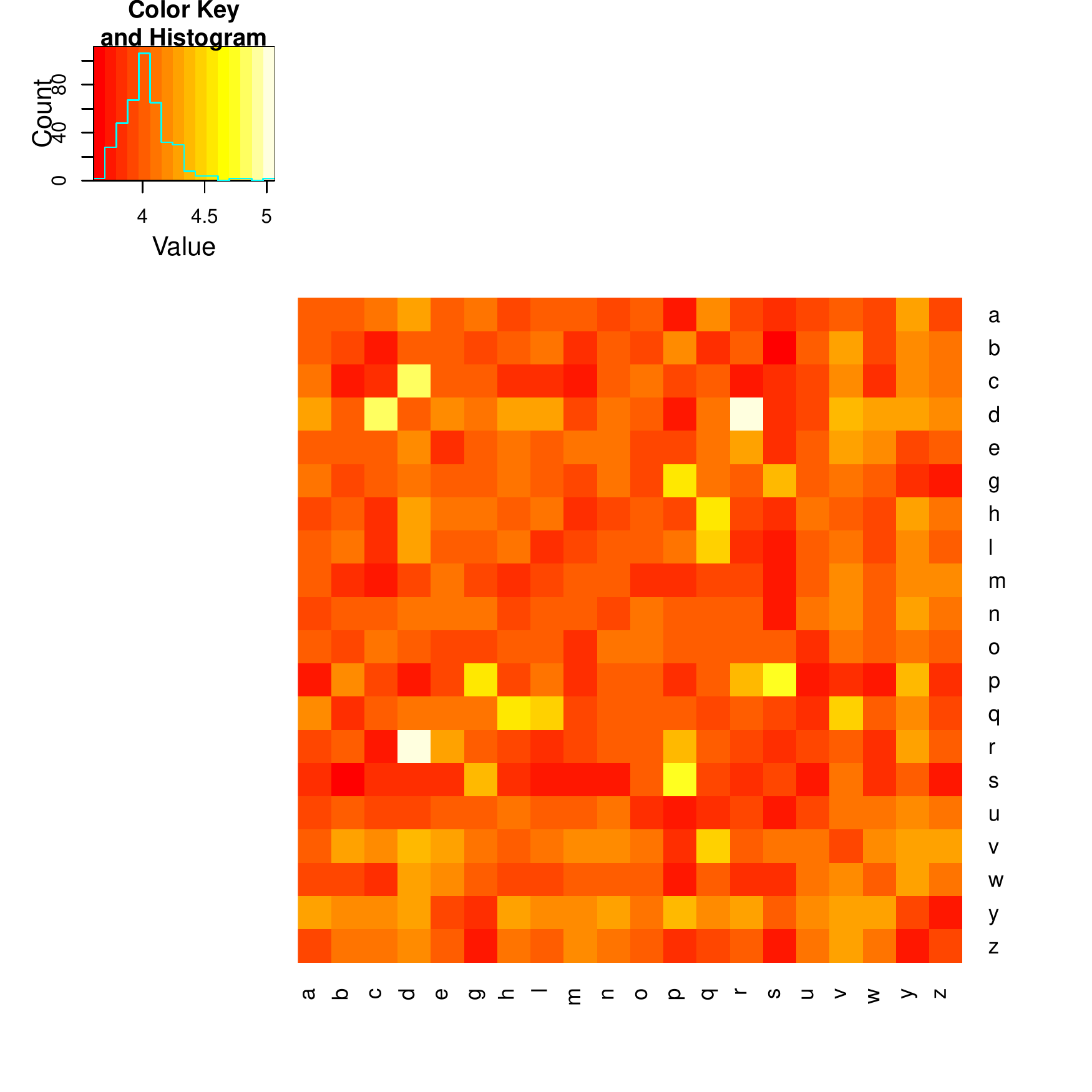}
\includegraphics[width=0.5\textwidth]{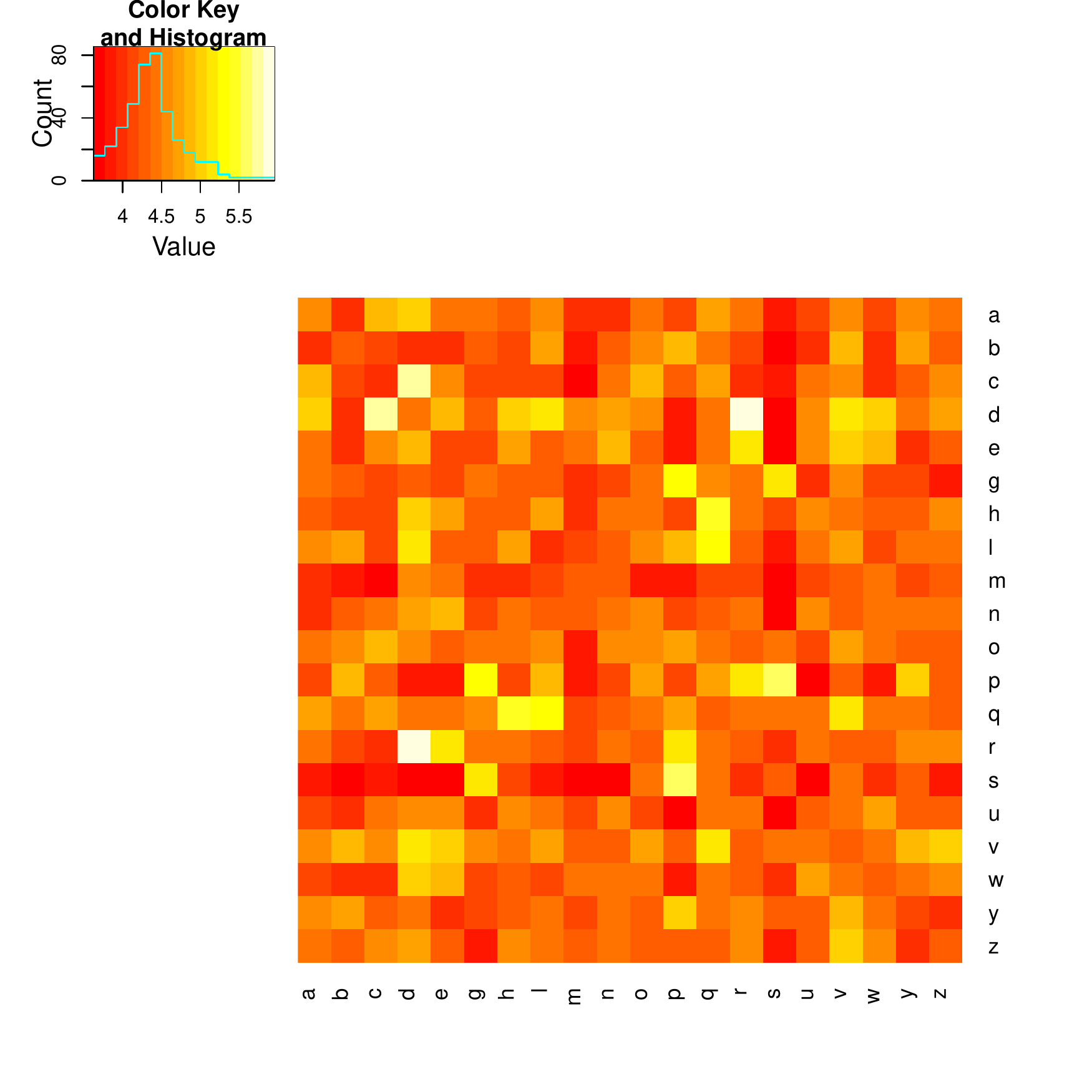}
\caption{Log of mean number of black-box calls for LMF (left) and Binary Search (right).}
\label{fig:valcomp_bbcallsheatmap}
\end{figure}

Finally, we give the average distance values on our benchmark set both under a fixed translation (specifically, with start points of $\pi$ and $\sigma$ normalized to the origin) and under translation in Figure~\ref{fig:valcomp_values}. Note that using naive approaches computing these tables would have been computationally extremely costly.

\begin{figure}
\includegraphics[width=0.5\textwidth]{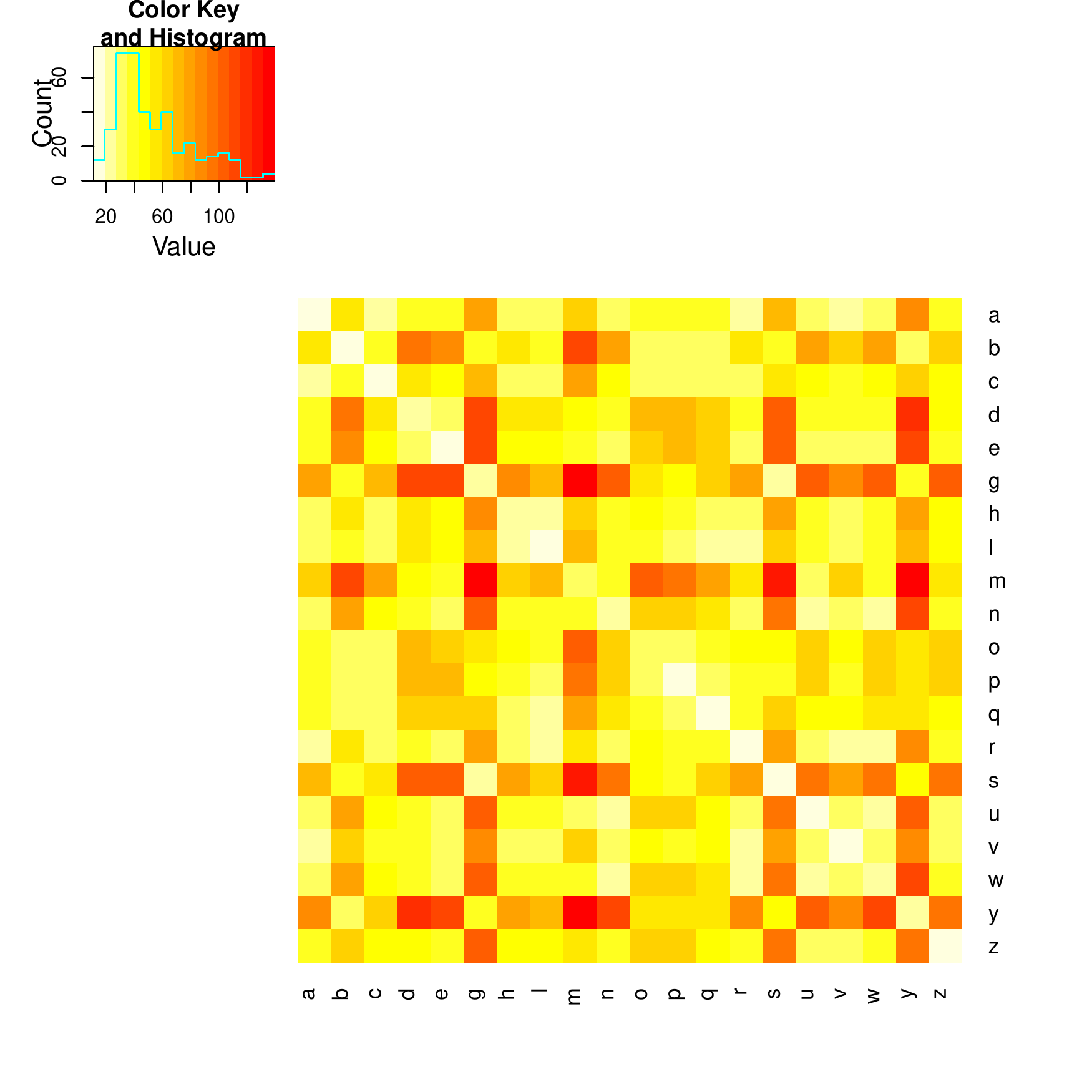}
\includegraphics[width=0.5\textwidth]{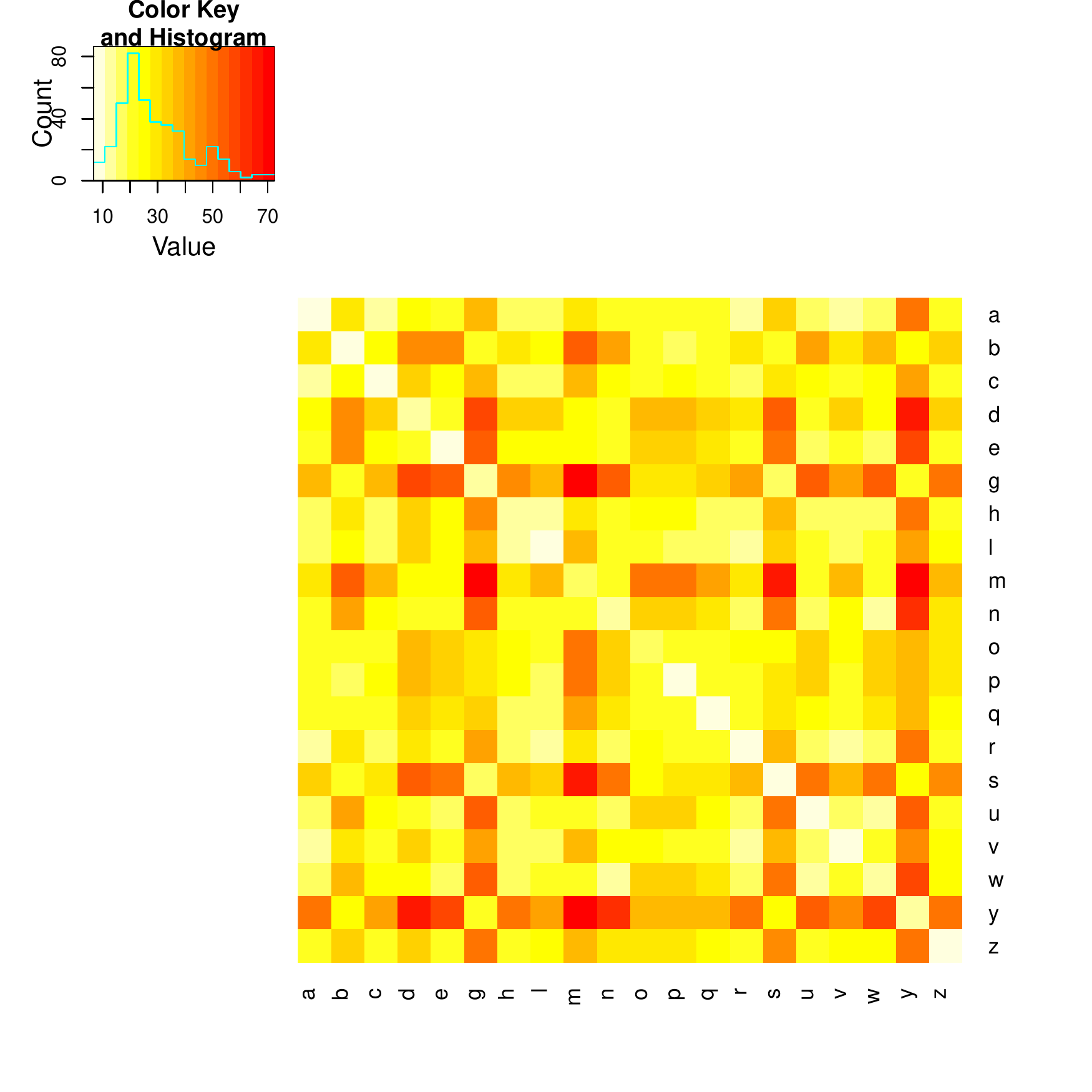}
\includegraphics[width=0.5\textwidth]{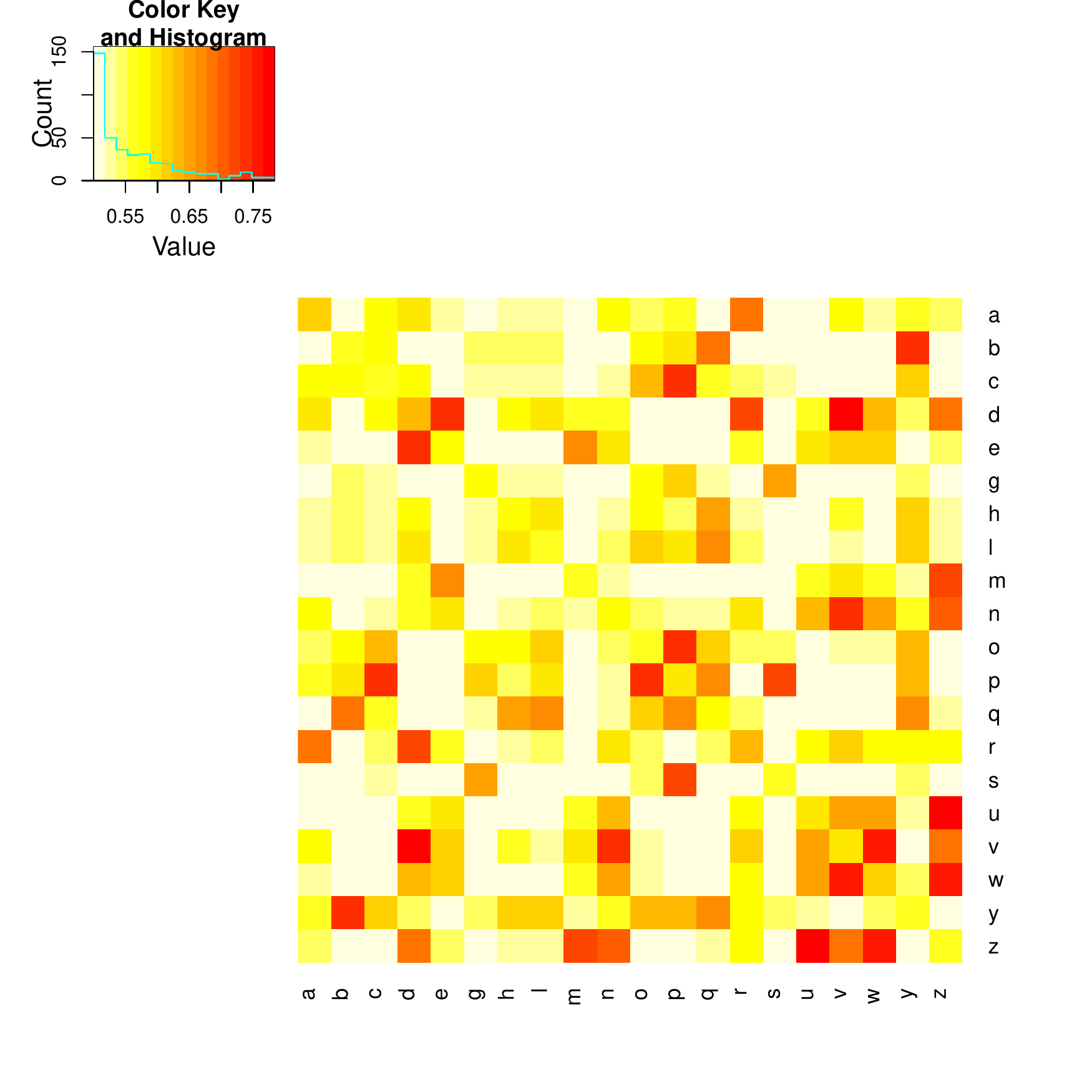}
\caption{Average Fréchet distance value (top left) and average Fréchet distance under translation value (top right), as well as the quotient of these values (bottom left).}
\label{fig:valcomp_values}
\end{figure}

\begin{table}
\caption{Statistics for approximate value computation for $\Nsamples = 100$. In parentheses, we give average values over the total of 21,000 curve pairs.}
\begin{tabular}{llrc}
\toprule
\textbf{Algorithm} & \multicolumn{2}{c}{\textbf{Time}} &  \textbf{Black-Box Calls}  \\
\midrule
LMF & \multicolumn{2}{c}{2,938,512 ms} & 260,128,449 \\
& \multicolumn{2}{c}{(140.0 ms per instance)} & (12,387.1 per instance)  \\
\cmidrule(r){2-3} 
& - Preprocessing & 71,728 ms  \\
\cmidrule(r){2-3}
& - Black-box calls (Lipschitz) & 400,189 ms  \\
\cmidrule(r){2-3}& - Arrangement estimation &166,479 ms  \\
\cmidrule(r){2-3}
& - Arrangement algorithm & 2,250,493 ms  \\
& \hphantom{bla} * Construction & 1,537,500 ms \\
& \hphantom{bla} * Black-box calls & 545,442 ms \\
\midrule 
Binary Search & \multicolumn{2}{c}{10,555,630 ms} & 875,424,988\\
& \multicolumn{2}{c}{( 502.7 ms per instance)} &  (41,686.9 per instance) \\
\bottomrule
 \end{tabular}
\label{table:totalvalCompBenchmarkN100}
\end{table}%

%% file: trunk/conclusion.tex
\section{Conclusion}
We engineer the first practical implementation for the discrete Fréchet distance under translation in the plane. While previous algorithmic solution for the problem solve it via expensive discrete methods, we introduce a new method from continuous optimization to achieve significant speed-ups on realistic inputs. 
  This is analogous to the success of integer programming solvers which, while optimizing a discrete problem, choose to work over the reals to gain access to linear programming relaxations, cutting planes methods, and more.   A novelty here is that we successfully apply  such methods to obtain drastic speed-ups for a \emph{polynomial-time problem}.

We leave as open problems to determine whether there are reasonable analogues of further ideas from integer programming, such as cutting plane methods or preconditioning, that could help to get further improved algorithms for our problem.
More generally, we believe that this gives an exciting direction for algorithm engineering in general that should find wider applications. 
A particular direction in this vein is the use of our methods to compute rotation- or scaling-invariant versions of the Fr\'echet distance. Intuitively, by introducing additional dimensions in our search space, our methods can in principle also be used to optimize over such additional degrees of freedom. However, the Lipschitz condition changes significantly,
and we leave it to future work to determine the applicability in these settings.
